\documentclass[lettersize,journal]{IEEEtran}
\usepackage{amsmath,amsfonts}
\usepackage{algorithm}
\usepackage{algpseudocode}
\usepackage{amssymb}
\usepackage{mathabx}
\usepackage{hyperref}
\usepackage{array}
\usepackage{textcomp}
\usepackage{stfloats}
\usepackage{url}
\usepackage{verbatim}
\usepackage{graphicx}
\usepackage{booktabs}
\usepackage{array}
\usepackage{xcolor}
\usepackage{cite}
\usepackage{multirow}
\usepackage{bm}

\usepackage{pgfplots}
\usepackage{tikz}
\usepackage{amsthm}
\newtheorem{theorem}{Theorem}
\newtheorem{definition}{Definition}

\newtheorem{question}{Question}
\newtheorem{lemma}{Lemma}

\DeclareMathOperator{\dom}{dom}
\DeclareMathOperator{\seq}{seq}
\DeclareMathOperator{\lh}{lh}
\DeclareMathOperator{\app}{app}
\DeclareMathOperator{\proj}{proj}
\DeclareMathOperator{\Input}{Input}
\DeclareMathOperator{\Output}{Output}
\DeclareMathOperator{\Com}{Com}
\DeclareMathOperator{\Rec}{Rec}
\DeclareMathOperator{\Min}{Min}
\DeclareMathOperator{\isRat}{isRat}
\DeclareMathOperator{\rat}{rat}
\DeclareMathOperator{\isCom}{isCom}
\DeclareMathOperator{\com}{com}
\DeclareMathOperator{\FindCode}{FindCode}
\DeclareMathOperator{\AchievesError}{AchievesError}
\DeclareMathOperator{\MessageNumber}{MessageNumber}
\DeclareMathOperator{\BLB}{BLB}
\DeclareMathOperator{\Codes}{Codes}
\DeclareMathOperator{\Code}{Code}
\DeclareMathOperator{\RowNumber}{RowNumber}
\DeclareMathOperator{\ColumnNumber}{ColumnNumber}
\DeclareMathOperator{\Element}{Element}
\DeclareMathOperator{\isCode}{isCode}
\DeclareMathOperator{\ParCodes}{ParCodes}
\DeclareMathOperator{\notInside}{notInside}

\DeclareMathOperator{\Kron}{Kron}
\DeclareMathOperator{\Row}{Row}
\DeclareMathOperator{\Col}{Col}
\DeclareMathOperator{\CodeLength}{CodeLength}
\DeclareMathOperator{\MinLength}{MinLength}
\DeclareMathOperator{\RLB}{RLB}
\DeclareMathOperator{\FindCodeExt}{FindCodeExt}
\DeclareMathOperator{\MessageNumberExt}{MessageNumberExt}
\DeclareMathOperator{\Capacity}{Capacity}
\DeclareMathOperator{\RatInterpolation}{RatInterpolation}
\DeclareMathOperator{\even}{even}
\DeclareMathOperator{\odd}{odd}

\definecolor{c11}{rgb}{0.2, 0.4, 0.6}  
\definecolor{c12}{rgb}{0.6, 0.4, 0.2}  
\definecolor{c13}{rgb}{0.4, 0.6, 0.2}  
\definecolor{c14}{rgb}{0.5, 0.3, 0.5}  
\definecolor{c15}{rgb}{0.6, 0.2, 0.2}  

\definecolor{c21}{rgb}{0.2, 0.4, 0.6}  
\definecolor{c22}{rgb}{0.6, 0.4, 0.2}  
\definecolor{c23}{rgb}{0.4, 0.6, 0.2}  
\definecolor{c24}{rgb}{0.5, 0.3, 0.5}  
\definecolor{c25}{rgb}{0.6, 0.2, 0.2}  

\definecolor{c31}{rgb}{0.0, 0.6, 1.0}
\definecolor{c32}{rgb}{0.1, 0.4, 0.7}
\definecolor{c33}{rgb}{0.2, 0.2, 0.4}

\begin{document}

\title{On the Computability of Finding Capacity-Achieving Codes}

\author{Angelos Gkekas, Nikos A. Mitsiou,~\IEEEmembership{Graduate Member,~IEEE}, \\Ioannis Souldatos, and George K. Karagiannidis,~\IEEEmembership{Fellow,~IEEE}

\thanks{Angelos Gkekas, Nikos A. Mitsiou, and George K. Karagiannidis are with the Department of Electrical and Computer Engineering, Aristotle University of Thessaloniki, 54124 Thessaloniki, Greece (e-mails: gkekasaa@ece.auth.gr, nmitsiou@auth.gr, geokarag@auth.gr).}
\thanks{Ioannis Souldatos is with the Department of Mathematics, Aristotle University of Thessaloniki, 54124 Thessaloniki, Greece (e-mail: souldatos@math.auth.gr).}

}
\vspace{-9mm}

\maketitle

\begin{abstract}
This work studies the problem of constructing capacity-achieving codes from an algorithmic perspective.
Specifically, we prove that there exists a Turing machine which, given a discrete memoryless channel \(p_{Y\mid X}\), a target rate \(R\) less than the channel capacity \(C(p_{Y\mid X})\), and an error tolerance \(\epsilon > 0\), outputs a block code \(\mathcal{C}\) achieving a rate at least \(R\) and a maximum block error probability below \(\epsilon\). The machine operates in the general case where all transition probabilities of \(p_{Y\mid X}\) are computable real numbers, and the parameters \(R\) and \(\epsilon\) are rational. The proof builds on Shannon’s channel coding theorem and relies on an exhaustive search approach that systematically enumerates all codes of increasing block length until a valid code is found. This construction is formalized using the theory of recursive functions, yielding a \(\mu\)-recursive function \(\FindCode : \mathbb{N}^3 \rightharpoonup \mathbb{N}\) that takes as input appropriate encodings of \(p_{Y\mid X}\), \(R\), and \(\epsilon\), and, whenever \(R < C(p_{Y\mid X})\), outputs an encoding of a valid code. By Kleene’s normal form theorem, which establishes the computational equivalence between Turing machines and \(\mu\)-recursive functions, we conclude that the problem is solvable by a Turing machine. This result can also be extended to the case where \(\epsilon\) is a computable real number, while we further discuss an analogous generalization of our analysis when \(R\) is computable as well. We note that the assumptions that the probabilities of \(p_{Y\mid X}\), as well as \(\epsilon\) and \(R\) are computable real numbers cannot be further weakened, since computable reals constitute the largest subset of \(\mathbb{R}\) representable by algorithmic means.
\end{abstract}






\section{Introduction}

Capacity-achieving codes constitute a cornerstone of modern communication theory. They enable reliable information transmission over noisy environments by ensuring that the probability of decoding error can be made arbitrarily small, while incurring only a bounded increase in codeword length. Specifically, for any discrete memoryless channel (DMC) without feedback, characterized by the conditional distribution \(p_{Y \mid X}\),  {there exists a fundamental limit \(C(p_{Y \mid X})\), referred to as the channel capacity.} Shannon’s channel coding theorem establishes that for any desired error tolerance \(\epsilon > 0\) and any coding rate \(R \in (0, C)\), there exists a coding scheme \(\mathcal{C}\) that simultaneously achieves error probability smaller than \(\epsilon\) and rate at least \(R\). This seminal result, first established in Shannon’s groundbreaking work \cite{shannon1948mathematical}, laid the foundation for the entire field of modern digital communications.

However, Shannon’s original proof did not provide a formal approach for designing such codes. Since then, a variety of explicit code constructions that approach or achieve capacity have been developed. Among the most widely studied are low-density parity-check (LDPC) codes \cite{gallager2003low}, Reed–Solomon codes \cite{reed1960polynomial}, turbo codes \cite{berrou1993near}, and polar codes \cite{arikan2009channel}. These codes have been extensively analyzed and successfully implemented in practical communication systems and standards, where they have demonstrated exceptional performance, underscoring the lasting importance of Shannon's original result. Nonetheless, each of these constructions is typically tailored to specific families of channel models. Naturally, this raises the question of whether constructing capacity-achieving codes for arbitrary DMCs is feasible from a computability perspective.

Among the most well-known tools for studying computability are Turing machines and \(\mu\)-recursive functions. Turing machines, introduced by Turing in his seminal work \cite{turing1936computable}, constitute the earliest and most widely used formal model of computation, and are regarded as the foundation of theoretical computer science. The class of \(\mu\)-recursive functions provides an alternative model, defined as a collection of partial functions from tuples over \(\mathbb{N}\) to \(\mathbb{N}\). A notable subclass is that of the primitive recursive functions, first introduced in Gödel’s proof of the incompleteness theorems \cite{godel1931formal}. Their first complete formulation was later developed in Kleene’s studies on recursion theory \cite{kleene1936general}. For a comprehensive exposition of \(\mu\)-recursive functions, see \cite{kleene1952introduction, odifreddi1992classical, rogers1987theory, soare1999recursively}. These two models of computation are equivalent, in the sense that any computational task performed in one model can be simulated in the other. This equivalence is established by Kleene’s normal form theorem \cite{kleene1952introduction}, which allows us to freely adopt either model in our analysis. 

\subsection{Literature Review}

There are only a handful of relevant contributions that examine the computability of constructing channel codes for DMCs. First, in \cite{boche2020computability}, the computability of the zero-error capacity  for DMCs was studied using Turing machines and Kolmogorov oracles and it was proved that the zero-error capacity function is semi-computable in this setting, while \cite{cheraghchi2009capacity} employed a probabilistic channel coding construction to obtain capacity-achieving linear codes for a broad class of channels. However, the most fundamental work on the computability of constructing channel codes is \cite{boche2022turing}. In \cite{boche2022turing}, the following question was addressed
\begin{question}\label{q:stronger_problem}
For a given computable family of channels \(\mathcal{W}\) and an error tolerance \(\epsilon \in (0,1) \cap \mathbb{Q}\), does there exist a Turing machine \(M_{\mathcal{W}, \epsilon}\) such that, for every DMC \(p_{Y \mid X} \in \mathcal{W}\) and every blocklength \(n\), the machine outputs a block code \(\mathcal{C}_n\) of length \(n\), whose maximum block error probability \(\lambda_n\) and rate \(R_n\) satisfy \(\lambda_n  \to 0\) and \(R_n \to C(p_{Y \mid X})\) as \(n \to \infty\)?
\end{question}
A negative answer to this question was provided using tools from the theory of Turing machines, recursive functions and computable real analysis. Moreover, this result also implies that that such an algorithmic construction remains impossible even when the optimality condition is dropped and codes only need to achieve a fraction of the capacity.


\subsection{Motivation and Contribution}

The negative answer to the problem formulation stated in \cite{boche2022turing} motivates the exploration of weaker assumptions for the formulation of the problem of constructing capacity-achieving codes, under which the problem is computable. Ideally, one is interested in finding the least weak assumptions for which this problem remains computable. As a first step in this direction, Question \ref{q:problem_informal} arises

\begin{question}\label{q:problem_informal}
Does there exist a Turing machine \(M\) such that, when given as input a DMC \(p_{Y \mid X}\), a rate \(R < C(p_{Y \mid X})\) and an error tolerance \(\epsilon > 0\), it outputs a block code \(\mathcal{C}\) with rate at least \(R\) and with maximum block error probability \(\lambda_{\max} < \epsilon\)?
\end{question}
We prove that the answer to this question is positive. Specifically, we show that the above problem is solvable by a Turing machine in the general setting where all parameters of the channel \(p_{Y \mid X}\) are arbitrary computable real numbers, while the error tolerance \(\epsilon\) and the rate \(R\) are rational. We also extend this result by allowing \(\epsilon\) to be any computable real number, while we further provide an informal discussion of a similar extension of \(R\) to computable real numbers too. To prove this, first, we describe a solution to the capacity-achieving code construction problem via an exhaustive search algorithm that relies on appropriately predefined \(\mu\)-recursive functions. We then formally construct a \(\mu\)-recursive function that implements this algorithm. Finally, we invoke Kleene's normal form theorem to conclude that the problem is solvable by a Turing machine. We note that the intermediate step of converting the algorithm into a \(\mu\)-recursive function ensures full rigor, since it removes ambiguities in the implementation of certain operations, such as the representation of computable reals or rationals of arbitrary precision and the arithmetic performed on them. By contrast, the final conversion from a \(\mu\)-recursive function to a Turing machine is primarily aesthetic, as Turing machines are the standard formalism in which computability results are typically expressed.

Interestingly, the positive answer to Question \ref{q:problem_informal} also implies the existence of a Turing machine that takes as input a DMC \(p_{Y \mid X}\) and a parameter \(k \in \mathbb{N}\), and computes a code \(\mathcal{C}_k\) that achieves a rate of at least \(C(p_{Y \mid X}) - \tfrac{1}{k}\) with a maximum block error probability below \(\tfrac{1}{k}\). Hence, the sequence of codes \(\{\mathcal{C}_k\}\) satisfies the desired asymptotic properties of the rate approaching the channel capacity and the error probability tending to zero. The detailed construction of this Turing machine is discussed in Appendix \ref{app:extension3}. 

At first glance, this may seem to contradict the negative answer to Question 1 by \cite{boche2022turing}. However, this is not the case for two main reasons. First, the machine proven impossible in \cite{boche2022turing} takes as input a DMC \(p_{Y \mid X}\) and a number \(n\), and outputs a code \(\mathcal{C}_n\) with blocklength exactly \(n\). In contrast, the machine proposed in this work takes as input a DMC \(p_{Y \mid X}\) and a number \(k\), and produces a code \(\mathcal{C}_k\) without any restriction on its blocklength. Second, \cite{boche2022turing} defines the general notion of \emph{computable families} of DMCs and prove their results under the assumption that the DMC input space is one such family. In our case, we consider arbitrary DMCs of any dimension, but we require that all transition probabilities \(p_{Y \mid X}\) are computable real numbers.
\section{Outline}
The remainder of this paper is organized as follows. Section \ref{sec:recursion_framework} introduces some key concepts and results from the theory of recursive functions. In Section \ref{sec:computability_framework}, we present the model of computable real numbers used throughout the paper. Section \ref{sec:solution} formulates the main problem and develops a recursive function that solves it. Finally, Section \ref{sec:conclusion} offers concluding remarks.

\section{Recursion Framework}\label{sec:recursion_framework}

In this section, we introduce some concepts from the theory of recursive functions that are fundamental to our analysis. These include the classes of primitive recursive and \(\mu\)-recursive functions, primitive recursive encodings, Turing machines, Kleene's normal form theorem, and a recursion framework for functions and relations of rational numbers. These concepts can be traced back to the works of Gödel \cite{godel1931formal}, Turing \cite{turing1936computable}, and Kleene \cite{kleene1936general, kleene1943recursive, kleene1938notation}. For a modern exposition, we refer the reader to \cite{rogers1987theory, soare1999recursively, odifreddi1992classical}.

\subsection{Primitive and \(\mu-\)Recursion}

We begin by introducing the concepts of primitive recursiveness, \(\mu\)-recursiveness and primitive recursive encodings. As a first step, we define the notion of a \emph{partial function}.

\begin{definition}[Partial Function]
A partial function \(f : X \rightharpoonup Y\) is a function \(f : S \to Y\) with \(S \subseteq X\). The set \(S\) is called the domain of \(f\) and it is denoted by \(\dom(f)\). If \(x \in S\), we write \(f(x) \downarrow\) and if \(x \in X \setminus S\) we write \(f(x) \uparrow\) or \(f(x) = \bot\). If \(S = X\), then \(f\) is called a total function.
\end{definition}

We will mainly work with partial functions of the form \(f : \mathbb{N}^n \rightharpoonup \mathbb{N}\) for \(n \in \mathbb{N}\). We now define the operations of \emph{composition}, \emph{primitive recursion} and \emph{minimization} for such functions.

\begin{definition}[Composition of Partial Functions]\label{def:composition}
Let \(f : \mathbb{N}^m \rightharpoonup \mathbb{N}\) and \(g_1, g_2, \dots g_m : \mathbb{N}^n \rightharpoonup \mathbb{N}\) be partial functions. The composition \(h(\bar{x}) = f(g_1(\bar{x}), g_2(\bar{x}), \dots g_m(\bar{x}))\) is a partial function \(h : \mathbb{N}^n \rightharpoonup \mathbb{N}\) defined as:
\begin{equation}
    h(\bar{x}) = \begin{cases}
    f(c_1, c_2, \dots c_m), &\text{if } g_i(\bar{x}) = c_i \neq \bot, \ \forall i \\ &\text{and } f(c_1, c_2, \dots c_m) \downarrow \\
    \bot, &\text{otherwise}
    \end{cases}
\end{equation}
\end{definition}

\begin{definition}[Primitive Recursion]
Let \(g : \mathbb{N}^n \rightharpoonup \mathbb{N}\) and \(h : \mathbb{N}^{n+2} \rightharpoonup \mathbb{N}\) be partial functions. It can be shown that there exists a unique partial function \(f : \mathbb{N}^{n+1} \rightharpoonup \mathbb{N}\) such that:
\begin{equation}
    \begin{cases}
    f(0, \bar{x}) = g(\bar{x}), &\forall{\bar{x}} \in \mathbb{N}^n \\
    f(y + 1, \bar{x}) = h(f(y, \bar{x}), y, \bar{x}), &\forall y \in \mathbb{N}, \forall\bar{x} \in \mathbb{N}^n
    \end{cases}
\end{equation}
The compositions in the above expression are interpreted in accordance with definition \ref{def:composition}. This means that:
\begin{equation}
    \begin{cases}
    f(0, \bar{x}) \downarrow \Leftrightarrow g(\bar{x}) \downarrow \\
    f(y + 1, \bar{x}) \downarrow \Leftrightarrow f(y, \bar{x}) = c \neq \bot \text{ and } h(c, y, \bar{x}) \downarrow
    \end{cases}
\end{equation}
We say that this function \(f\) is defined by primitive recursion on \(h\) with base case \(g\).
\end{definition}
\begin{definition}[Minimization]
Let \(g : \mathbb{N}^{n+1} \rightharpoonup \mathbb{N}\) be a partial function. Define \(D_{x, \bar{y}} = \{i \in \mathbb{N} \mid i \geq x,\ g(i, \bar{y}) = 0 \text{ and } g(j, \bar{y}) \downarrow \text{ for all } j \text{ with } x \leq j \leq i\} \). The minimization of \(g\) is defined as the partial function \(f : \mathbb{N}^{n+1} \rightharpoonup \mathbb{N}\) with \(f(x, \bar{y}) = \mu i \geq x : (g(i, \bar{y}) = 0)\) where:
\begin{equation}
    \mu i \geq x : (g(i, \bar{y}) = 0) = \begin{cases}
    \min D_{x, \bar{y}}, &\text{if } D_{x, \bar{y}} \neq \varnothing \\
    \bot, &\text{otherwise} 
    \end{cases}
\end{equation}
If \(x=0\) we write for simplicity \(\mu i :(g(i, \bar{y}) = 0)\) instead of \(\mu i \geq 0 :(g(i, \bar{y}) = 0)\).
\end{definition}

We now proceed to define the classes of \emph{primitive recursive} and \emph{\(\mu-\)recursive functions}. To do so, we first introduce the set of \emph{basic functions} from which these classes are constructed.

\begin{definition}[Basic Functions]
The set \(B\) of basic functions is the set that includes exactly the following functions:
\begin{enumerate}
    \item For \(n,k \in \mathbb{N}\) the constant functions \(C^n_k : \mathbb{N}^n \rightarrow \mathbb{N}\) with \(C^n_k(\bar{x}) = k\)
    \item For \(n \in \mathbb{N}\) and \(0 \leq i < n\) the projections \(P^n_i : \mathbb{N}^n \rightarrow \mathbb{N}\) with \(P^n_i(x_0, x_1, \dots x_{n-1}) = x_i\)
    \item The successor function \(S(x) = x+1\)
\end{enumerate}
\end{definition}

\begin{definition}[Primitive Recursive and \(\mu-\)Recursive Functions]
The class \(R_p\) of primitive recursive functions is the closure of \(B\) for composition and primitive recursion. The class \(R_\mu\) of \(\mu-\)recursive functions is the closure of \(B\) for composition, primitive recursion and minimization.
\end{definition}

The notions of primitive and \(\mu-\)recursiveness also extend to relations.

\begin{definition}[Primitive Recursive and \(\mu-\)Recursive Relations] The characteristic function of a relation \(R \subseteq \mathbb{N}^n\) is the function \(\chi_R : \mathbb{N}^n \rightarrow \mathbb{N}\) defined as:
\begin{equation}
\chi_R(\bar{x}) = \begin{cases}
1, &\text{if } R(\bar{x}) \\
0, &\text{otherwise}
\end{cases}
\end{equation}
A relation \(R\) is called primitive recursive (\(\mu-\)recursive) iff \(\chi_R \in R_p\) (iff \(\chi_R \in R_{\mu}\)).
\end{definition}

Furthermore, we introduce the concept of \emph{primitive recursive encodings}.

\begin{definition}[Primitive Recursive Encoding]
We denote by \(\mathbb{N}^* = \bigcup_{n \in \mathbb{N}} \mathbb{N}^n\) the set of all finite sequences of natural numbers, including the empty sequence \(\varepsilon\). A function \(\langle \rangle : \mathbb{N}^* \rightarrow \mathbb{N}\) is called a primitive recursive encoding iff it satisfies the following conditions:
\begin{enumerate}
    \item \(\langle \rangle\) is injective 
    \item The relation \(\seq \subseteq \mathbb{N}\) defined by:
    \begin{equation}
        \seq(u) \Leftrightarrow \exists \bar{x} \in \mathbb{N}^* : \langle \bar{x} \rangle = u
    \end{equation}
    is primitive recursive
    \item The functions \(F_n : \mathbb{N}^n \rightarrow \mathbb{N}\) defined by:
    \begin{equation}
        F_n(\bar{x}) = \langle \bar{x} \rangle
    \end{equation}
    are primitive recursive for all \(n \in \mathbb{N}\)
    \item \(\langle x_0, x_1, \dots x_{n-1} \rangle > x_i\) for all \((x_0, x_1, \dots x_{n-1}) \in \mathbb{N}^*\) and \(0 \leq i < n\)
    \item There exist primitive recursive functions \(\lh : \mathbb{N} \rightarrow \mathbb{N}\) and \(\app\), \(\proj : \mathbb{N}^2 \rightarrow \mathbb{N}\) such that for all \((x_0, x_1, \dots x_{n-1}) \in \mathbb{N}^*\), \(0 \leq i < n\) and \(y \in \mathbb{N}\):
    \begin{itemize}
        \item \(\lh(\langle x_0, x_1, \dots x_{n-1} \rangle) = n\)
        \item \(\app(\langle x_0, x_1, \dots x_{n-1} \rangle, y) = \langle x_0, x_1, \dots x_{n-1}, y \rangle\)
        \item \(\proj(\langle x_0, x_1, \dots x_{n-1} \rangle, i) = x_i\)
    \end{itemize}
\end{enumerate}
\end{definition}

In general, we are only concerned with the values of the functions \(\lh\), \(\app\) and \(\proj\) when their arguments are of the form specified in condition (5) of the preceding definition. Although these functions are also well-defined for inputs not of this form, their values in these cases are not relevant.

For a given primitive recursive encoding we will use the notation \((u)_{i_1, i_2, \dots i_k}\) to denote nested application of the function \(\proj\) as:
\begin{equation}
    (u)_{i_1, i_2, \dots i_k} = \proj(\dots \proj(\proj(u, i_1), i_2) \dots , i_k)
\end{equation}

An example of a primitive recursive encoding is the \emph{classical encoding} defined by:
\begin{equation}
    \begin{cases}
        \langle \varepsilon \rangle = 1\\
        \langle x_0, x_1, \dots x_{n-1} \rangle = p_0^{x_0+1}\cdot p_1^{x_1 + 1} \cdot \dots p_{n-1}^{x_{n-1} +1}
    \end{cases}
\end{equation}
where \(p_i\) denotes the \(i-\)th prime number, starting with \(p_0 = 2\).

We list some standard lemmata concerning the closure properties of the classes \(R_p\) and \(R_\mu\), which will be used in our analysis. Proofs of these results can be found in \cite{rogers1987theory, soare1999recursively, odifreddi1992classical}.

\begin{lemma}[Standard Recursive Functions]
The standard addition, multiplication and exponentiation over \(\mathbb{N}\) are primitive recursive functions, as is the subtraction \(\dotdiv : \mathbb{N}^2 \rightarrow \mathbb{N}\) over \(\mathbb{N}\) defined by:
\begin{equation}
    x \dotdiv y = \begin{cases}
        x - y, &\text{if } x \geq y \\
        0, &\text{otherwise}
    \end{cases}
\end{equation}
\end{lemma}

\begin{lemma}[Definition by Cases]
Let \(R_1, R_2, \dots R_m \subseteq \mathbb{N}^n\) be primitive recursive (\(\mu-\)recursive) relations that partition \(\mathbb{N}^n\) and \(g_1, g_2, \dots g_m : \mathbb{N}^n \rightharpoonup \mathbb{N}\) be primitive recursive (\(\mu-\)recursive) functions. Then the function \(f : \mathbb{N}^n \rightharpoonup \mathbb{N}\) defined by:
\begin{equation}
    f(\bar{x}) = \begin{cases}
        g_1(\bar{x}), &\text{if } R_1(\bar{x}) \\
        g_2(\bar{x}), &\text{if } R_2(\bar{x}) \\
        \quad\vdots \\
        g_m(\bar{x}), &\text{if } R_m(\bar{x})
    \end{cases}
\end{equation}
is also primitive recursive (\(\mu-\)recursive).
\end{lemma}

\begin{lemma}[Bounded Minimization]
Let \(R \subseteq \mathbb{N}^{n+1}\) be a relation. Define \(D_{x,\bar{y}} = \{i \in \mathbb{N} \mid i \leq x \text{ and } R(i, \bar{y})\}\). The bounded minimization of \(R\) is defined as the function \(f : \mathbb{N}^{n+1} \rightharpoonup \mathbb{N}\) with \(f(x, \bar{y}) = \mu i \leq x : (R(i, \bar{y}))\) where:
\begin{equation}
    \mu i \leq x : (R(i, \bar{y})) = \begin{cases}
        \min D_{x, \bar{y}}, &\text{if } D_{x,\bar{y}} \neq \varnothing \\
        x+1, &\text{otherwise}
    \end{cases}
\end{equation}
If the relation \(R\) is primitive recursive (\(\mu\)-recursive), then the function \(f(x, \bar{y}) = \mu i \leq x : (R(i, \bar{y}))\) is also primitive recursive (\(\mu\)-recursive).
\end{lemma}

\begin{lemma}[Standard Recursive Relations]
The relations \(=, \leq, \geq, <, > \subseteq \mathbb{N}^2\) are primitive recursive.
\end{lemma}

\begin{lemma}[Closure for Logical Connectives]
If the relations \(P, Q \subseteq \mathbb{N}^2\) are primitive recursive (\(\mu-\)recursive), then so are the relations \(P \land Q\), \(P \lor Q\), \(P \rightarrow Q\) and \(\neg P\).
\end{lemma}

\begin{lemma}[Closure for Bounded Quantifiers]
Let \(R \subseteq \mathbb{N}^{n+1}\) be a primitive recursive (\(\mu\)-recursive) relation. Then the relations:
\begin{gather}
    Q(z, \bar{y}) = \forall x \leq z : (R(x, \bar{y})) \\
    P(z, \bar{y}) = \exists x \leq z : (R(x, \bar{y}))
\end{gather}
are also primitive recursive (\(\mu\)-recursive).
\end{lemma}

\begin{lemma}[Minimization of Relations]
Let \(R \subseteq \mathbb{N}^{n+1}\) be a relation. Define \(D_{x,\bar{y}} = \{i \in \mathbb{N} \mid i \geq x \text{ and } R(i, \bar{y})\}\). The minimization of \(R\) is defined as the function \(f : \mathbb{N}^{n+1} \rightharpoonup \mathbb{N}\) with \(f(x, \bar{y}) = \mu i \geq x : (R(i, \bar{y}))\) where:
\begin{equation}
    \mu i \geq x : (R(i, \bar{y})) = \begin{cases}
        \min D_{x, \bar{y}}, &\text{if } D_{x,\bar{y}} \neq \varnothing \\
        \bot, &\text{otherwise}
    \end{cases}
\end{equation}
If \(x=0\) we simply write \(\mu i :(R(i, \bar{y}))\) instead of \(\mu i \geq x :(R(i, \bar{y}))\).
If the relation \(R\) is \(\mu\)-recursive, then the function \(f(x, \bar{y}) = \mu i \geq x : (R(i, \bar{y}))\) is also \(\mu\)-recursive.
\end{lemma}

\begin{lemma}[Concatenation]
    Let \(\langle \rangle : \mathbb{N}^* \rightarrow \mathbb{N}\) be a primitive recursive encoding. There exists a primitive recursive function \(* : \mathbb{N}^2 \rightarrow \mathbb{N}\), called concatenation, such that for all \(\bar{x}, \bar{y} \in \mathbb{N}^*\) we have:
    \begin{equation}
        \langle \bar{x} \rangle * \langle \bar{y} \rangle = \langle \bar{x}, \bar{y} \rangle
    \end{equation}
\end{lemma}

\subsection{General Recursion}

We now present the notion of general recursive functions. Informally, a function is said to be general recursive if it can be computed by an algorithm or by a computational machine. The most widely used formal model of computation is the \emph{Turing machine}, introduced by Alan Turing in his seminal work on the Entscheidungsproblem \cite{turing1936computable}. Accordingly, we define general recursive functions with respect to this model.

Given a Turing machine \(M = (Q, \Gamma, \square, \Sigma, \delta, q_0, F)\), where \(Q\) is the set of states, \(\Gamma\) and \(\Sigma\) are the tape and the input alphabet respectively, \(\square \in \Gamma\) is the blank symbol, \(\delta : Q \times \Gamma \rightharpoonup Q \times \Gamma \times \{\text{Left, Right}\}\) is the transition function, \(q_0 \in Q\) is the starting state and \(F \subseteq Q\) is the set of terminating states, we consider the set \(K\) of total states as the set of pairs \(C = (\sigma, q)\) where:
\begin{enumerate}
    \item \(\sigma : \mathbb{Z} \rightarrow \Gamma\) is a function that encodes the tape content, with \(\sigma(0)\) representing the symbol under the machine’s head
    \item \(q \in Q\) is the current state of the machine
\end{enumerate}

We will use the notation \(\sigma_{w}\), for \(w = w_0w_1\dots w_{L-1} \in \Gamma^*\) to denote the function from \(\sigma_w : \mathbb{Z} \rightarrow \Gamma\) defined by:
\begin{equation}
    \sigma_w(n) = \begin{cases}
    w_n, &\text{if } 0 \leq n < L \\
    \square, &\text{otherwise}
    \end{cases}
\end{equation}
That is, \(\sigma_w\) is the function that describes the content of the tape of \(M\) when the word \(w\) is written on it and the head of \(M\) is positioned over the first symbol of \(w\).

The computation relation \(\vdash_M^* \subseteq K \times K\) is defined by:
\begin{align*}
    C_1 \vdash_M^* C_2 \Leftrightarrow \ &  \text{there exists a computation of } M\\ &\text{that starts in the total state } C_1 \text{ and} \\ &\text{terminates in the total state } C_2
\end{align*}

Using the Turing machine model, we can define the class of \emph{general recursive partial functions}. Fix an input alphabet \(\Sigma_R\), a tape alphabet \(\Gamma_R\), and two injective functions:
\[
\Input : \mathbb{N}^* \rightarrow \Sigma_R^*, \quad \Output : \Gamma_R^* \rightarrow \mathbb{N}
\]
The function \(\Input\) encodes inputs \(\bar{x} \in \mathbb{N}^*\) as words over \(\Sigma_R\) and the function \(\Output\) decodes words over \(\Gamma_R\) into natural numbers.

\begin{definition}[General Recursive Function]
    A partial function \(f : \mathbb{N}^n \rightharpoonup \mathbb{N}\) is called general recursive iff there exists a Turing machine \(M = (Q, \Gamma_R, \square, \Sigma_R, \delta, q_0, F)\), with \(q_0 = n\), such that for every input \(\bar{x} \in \mathbb{N}^n\) the following conditions hold:
    \begin{gather}
        \begin{aligned}
        \bar{x} \in \dom(&f) \Rightarrow \\
        &\left(\sigma_{\Input(\bar{x})}, q_0\right) \vdash^*_M \left(\sigma_w, q_f \right) \text{ for some } q_f \in F \\
        &\text{ and } w \in \Gamma_R^* \text{ with } \Output(w) = f(\bar{x})
        \end{aligned}\\
        \begin{aligned}
        \bar{x} &\not\in \dom(f) \Rightarrow \\
        &\left(\sigma_{\Input(\bar{x})}, q_0\right) \not\vdash^*_M \left(\sigma_w, q_f \right) \text{ for any pair } (q_f, w) \in F \times \Gamma^*_R
        \end{aligned}
    \end{gather}
    If both of the above conditions hold, we say that the Turing machine \(M\) computes the partial function \(f\). The class of all general recursive functions is denoted by \(R(\mathbb{N})\).
\end{definition}

The restriction \(q_0 = n\) in the above definition ensures that a given Turing machine does not compute two distinct partial functions of different arities, with the exception of the \emph{empty functions} \(\epsilon_n : \mathbb{N}^n \rightharpoonup \mathbb{N}\) and \(\epsilon_m : \mathbb{N}^m \rightharpoonup \mathbb{N}\), defined by \(\dom(\epsilon_n) = \dom(\epsilon_m) = \varnothing\) for \(n \neq m\).

Without this restriction, it would be possible to construct a single machine computing multiple functions of different arities. For instance, a ``sum'' machine could compute both the binary function \(f : \mathbb{N}^2 \to \mathbb{N}\) defined by \(f(x, y) = x + y\) and the ternary function \(g : \mathbb{N}^3 \to \mathbb{N}\) defined by \(g(x, y, z) = x + y + z\).

\subsection{Normal Form Theorem}

A fundamental result in the theory of recursive functions is \emph{Kleene’s normal form theorem}. The first complete and formal statement, along with a detailed proof, is presented in his book \cite{kleene1952introduction}.

\begin{theorem}[Normal Form Theorem]\label{theorem:NF}
There exists a primitive recursive function \( U : \mathbb{N} \rightarrow \mathbb{N} \), primitive recursive relations \( T_n \subseteq \mathbb{N}^{n+2} \) for every \( n \in \mathbb{N} \) and injective primitive recursive functions \( S^m_n : \mathbb{N}^{m+1} \rightarrow \mathbb{N} \) for all \( n, m \in \mathbb{N} \), such that the following conditions hold:
\begin{enumerate}
    \item A partial function \( f : \mathbb{N}^n \rightharpoonup \mathbb{N} \) is general recursive iff there exists \( e \in \mathbb{N} \) such that:
    \begin{equation}
    f(\bar{x}) = U(\mu y : (T_n(y, e, \bar{x}))), \quad \forall \bar{x} \in \mathbb{N}^n
    \end{equation}
    If the above condition holds then \(e\) is called a code of \(f\).
    \item For every \( e\in \mathbb{N}, \bar{z} \in \mathbb{N}^m \) and \( \bar{x} \in \mathbb{N}^n \):
    \begin{equation}
    \begin{aligned}
    U(\mu y : (T_{m+n}&(y, e, \bar{z}, \bar{x}))) =\\
    &= U(\mu y : (T_{n}(y, S^m_n(e, \bar{z}), \bar{x})))
    \end{aligned}
    \end{equation}
\end{enumerate}
\end{theorem}

From Theorem \ref{theorem:NF} it follows that \(R(\mathbb{N}) \subseteq R_\mu\), since each \(f \in R(\mathbb{N})\) can be constructed by composing \(U\) with a minimization of \(T_n\), where \(n\) is equal to the arity of \(f\). The inverse inclusion \(R_\mu \subseteq R(\mathbb{N})\) can also be easily proven. Therefore we have \(R(\mathbb{N}) = R_\mu\). This allows us to use the term \emph{recursive function/relation} to refer to both \(\mu\)-recursive and general recursive partial functions/relations, as they are equivalent.

Theorem \ref{theorem:NF} allows us to define the \emph{universal recursive functions}.

\begin{definition}[Universal Recursive Function]
    The universal recursive functions \(\varphi^n : \mathbb{N}^{n+1} \rightharpoonup \mathbb{N}\) are defined for each \(n \in \mathbb{N}\) as:
    \begin{equation}
        \varphi^n(e, \bar{x}) = U(\mu y : (T_n(y, e, \bar{x})))
    \end{equation}
    For a given \(n \in \mathbb{N}\) every recursive function \(f : \mathbb{N}^n \rightharpoonup \mathbb{N}\) can be expressed as \(f(\bar{x}) = \varphi^n(e, \bar{x})\), where \(e\) is a code of \(f\).
\end{definition}

The \(S^m_n\) functions can be used to construct, using only primitive recursion, codes of arbitrarily complex recursive functions. This is expressed in the following lemma, the proof of which heavily relies on the \(S^m_n\) functions:

\begin{lemma}[Effectiveness of Composition, Primitive Recursion and Minimization]\label{lemma:effectiveness}
For every \(m, n \in \mathbb{N}\) there exist primitive recursive functions \(\Com^m_n : \mathbb{N}^{m+1} \rightarrow \mathbb{N}\), \(\Rec_n : \mathbb{N}^2 \rightarrow \mathbb{N}\) and \(\Min_n : \mathbb{N} \rightarrow \mathbb{N}\) such that the following conditions hold:
\begin{enumerate}
    \item If \(e_1, e_2, \dots e_m\) are codes of \(f_1, f_2, \dots f_m : \mathbb{N}^n \rightharpoonup \mathbb{N}\) and \(e_g\) is a code of \(g : \mathbb{N}^m \rightharpoonup \mathbb{N}\) then \(\Com^m_n(e_g, e_1, e_2, \dots e_m)\) is a code of the composition \(g(f_1(\bar{x}), f_2(\bar{x}), \dots f_m(\bar{x}))\).
        
    \item If \(e_g\) is a code of \(g : \mathbb{N}^n \rightharpoonup \mathbb{N}\) and \(e_h\) is a code of \(h : \mathbb{N}^{n+2} \rightharpoonup \mathbb{N}\) then \(\Rec_n(e_g, e_h)\) is a code of the function \(f : \mathbb{N}^{n+1} \rightharpoonup \mathbb{N}\), which is defined by primitive recursion on \(h\) with base case \(g\).
        
    \item If \(e_g\) is a code of \(g : \mathbb{N}^{n+1} \rightharpoonup \mathbb{N}\) then \(\Min_n(e_g)\) is a code of the minimization \(\mu i \geq x : (g(i, \bar{y}) = 0)\).
\end{enumerate}
\end{lemma}

Lemma \ref{lemma:effectiveness} enables the construction of primitive recursive functions that compute codes for arbitrarily complex recursive functions. Given a set of base recursive functions \(f_1, f_2, \dots f_k\) with corresponding codes \(e_1, e_2, \dots e_k\), and a function \(g\) defined by a known sequence of compositions, primitive recursions, and minimizations on these base functions, we can easily define a primitive recursive function \(F : \mathbb{N}^k \rightarrow \mathbb{N}\) such that \(F(e_1, e_2, \dots e_k)\) is a code of \(g\). The function \(F\) is constructed by appropriately composing the functions \(\Com^m_n\), \(\Rec_n\), and \(\Min_n\) according to the sequence of operations that defines \(g\).

\subsection{Recursive Functions over the Rational Numbers}\label{subsec:rationals}

We extend the notions of primitive and general recursive functions to the set of rational numbers, enabling the application of results from recursion theory to model algorithms that operate on rational inputs. To this end, we fix a primitive recursive encoding \(\langle \rangle : \mathbb{N}^* \rightarrow \mathbb{N}\) for the remainder of this article.

We begin by defining an encoding of the rational numbers into the natural numbers. Using this encoding, any rational \(q \in \mathbb{Q}\) can be identified with a corresponding code \(x \in \mathbb{N}\). This identification, together with recursive functions on natural numbers, allows us to define recursive functions of the form \(f : \mathbb{N} \rightharpoonup \mathbb{Q}\) and \(g : \mathbb{Q}^n \to \mathbb{Q}\).

\begin{definition}[Encoding of Rationals]\label{def:rationals}
Let \(q = (-1)^s \frac{N}{D}\) be a rational number, where \(s, N, D \in \mathbb{N}\) and \(D \neq 0\). We will call \(\langle s, N, D \rangle \in \mathbb{N}\) a code of \(q\).  

We introduce the primitive recursive functions \(s_\mathbb{Q}\), \(N_\mathbb{Q}\), \(D_\mathbb{Q} : \mathbb{N} \rightarrow \mathbb{N}\), which satisfy:
\begin{equation}
    q = (-1)^{s_\mathbb{Q}(u)} \cdot \frac {N_\mathbb{Q}(u)}{D_\mathbb{Q}(u)}
\end{equation}
whenever \(u\) is a code of \(q\), by:
\begin{gather}
    s_\mathbb{Q}(u) = (u)_0 \\
    N_\mathbb{Q}(u) = (u)_1 \\
    D_\mathbb{Q}(u) = (u)_2
\end{gather}

We also define the primitive recursive relation \(\isRat \subseteq \mathbb{N}\) by:
\begin{equation}
    \begin{aligned}
        \isRat(u) &\Leftrightarrow u \text{ is a code of some rational } q \\
        &\Leftrightarrow \seq(u) \land \lh(u) = 3 \land (u)_2 \neq 0
    \end{aligned}
\end{equation}

Finally, we define the function \(\rat : \isRat \rightarrow \mathbb{Q}\), with \(\rat(u) = q\) whenever \(u\) is a code of \(q\), by:
\begin{equation*}
    \rat(u) = (-1)^{s_\mathbb{Q}(u)} \cdot \frac{N_\mathbb{Q}(u)}{D_\mathbb{Q}(u)}
\end{equation*}
\end{definition}

Note that a rational number \(q\) does not admit a unique code. For example, both \(\langle 0, 5, 2 \rangle\) and \(\langle 2, 10, 4 \rangle\) are codes of \(\frac{5}{2}\). Furthermore, the functions \(s_\mathbb{Q}\), \(N_\mathbb{Q}\), and \(D_\mathbb{Q}\) are defined on all inputs, including those for which \(\neg\isRat(u)\). However, their values in such cases are not meaningful for the encoding.

Definition \ref{def:rationals} allows us to represent functions and relations on rational numbers using recursive functions and relations on natural numbers. In particular, the four basic arithmetic operations, as well as the relations of equality, strict order, and non-strict order over the rationals, are all primitive recursive, as formalized in Lemma \ref{lemma:rational_operations}.

\begin{lemma}\label{lemma:rational_operations}
    There exist primitive recursive functions \(+_\mathbb{Q}\), \(-_\mathbb{Q}\), \(\cdot_\mathbb{Q}\), \(/_\mathbb{Q} : \mathbb{N}^2 \rightarrow \mathbb{N}\) such that for all \(x,y \in \mathbb{N}\) with \(\isRat(x)\) and \(\isRat(y)\) the following conditions hold:
    \begin{gather}
        \rat(x +_\mathbb{Q}y) = \rat(x) + \rat(y) \\
        \rat(x -_\mathbb{Q}y) = \rat(x) - \rat(y) \\
        \rat(x \cdot_\mathbb{Q}y) = \rat(x) \cdot \rat(y) \\
        \rat(x /_\mathbb{Q}y) = \begin{cases}
        \rat(x) / \rat(y), &\text{if } \rat(y) \neq 0\\
        0, &\text{if } \rat(y) = 0
        \end{cases}
    \end{gather}

    Furthermore, the relations \(=_\mathbb{Q}\), \(\leq_\mathbb{Q}\), \(<_\mathbb{Q} \subseteq \mathbb{N}^2\) defined by:
    \begin{gather}
    x =_\mathbb{Q} y \Leftrightarrow \isRat(x) \land \isRat(y) \land \rat(x) = \rat(y) \label{eq:rat=}\\
    x \leq_\mathbb{Q} y \Leftrightarrow \isRat(x) \land \isRat(y) \land \rat(x) \leq \rat(y) \label{eq:ratleq}\\
    x <_\mathbb{Q} y \Leftrightarrow \isRat(x) \land \isRat(y) \land \rat(x) < \rat(y) \label{eq:rat<}
    \end{gather}
    are primitive recursive.
\end{lemma}

\begin{proof}
We will show the existence of \(+_\mathbb{Q}\) and \(/_\mathbb{Q}\). The existence of \(-_\mathbb{Q}\) and \(\cdot_\mathbb{Q}\) can be established in a similar fashion. Let \(x,y \in \mathbb{N}\) with \(\isRat(x)\) and \(\isRat(y)\). We have:
\begin{align}
    \rat(x) &+ \rat(y) =\\
    & \frac{(-1)^{s_\mathbb{Q}(x)}N_\mathbb{Q}(x)D_\mathbb{Q}(y) + (-1)^{s_\mathbb{Q}(y)}N_\mathbb{Q}(y)D_\mathbb{Q}(x)}{D_\mathbb{Q}(x)D_\mathbb{Q}(y)} \label{eq:add_rat}
\end{align}
We consider three cases for expression \eqref{eq:add_rat}:
\begin{enumerate}
    \item if \(s_\mathbb{Q}(x) \equiv s_\mathbb{Q}(y) \pmod{2}\) then \(\rat(x)\) and \(\rat(y)\) have the same sign, and expression \eqref{eq:add_rat} can be written as:
    \begin{equation}
        (-1)^{s_\mathbb{Q}(x)}\frac{N_\mathbb{Q}(x)D_\mathbb{Q}(y) + N_\mathbb{Q}(y)D_\mathbb{Q}(x)}{D_\mathbb{Q}(x)D_\mathbb{Q}(y)}
    \end{equation}

    \item if \(s_\mathbb{Q}(x) \not\equiv s_\mathbb{Q}(y) \pmod{2}\) and \(N_\mathbb{Q}(x)D_\mathbb{Q}(y) \geq N_\mathbb{Q}(y)D_\mathbb{Q}(x)\) then expression \eqref{eq:add_rat} can be written as:
    \begin{equation}
        (-1)^{s_\mathbb{Q}(x)}\frac{N_\mathbb{Q}(x)D_\mathbb{Q}(y) \dotdiv N_\mathbb{Q}(y)D_\mathbb{Q}(x)}{D_\mathbb{Q}(x)D_\mathbb{Q}(y)}
    \end{equation}

    \item if \(s_\mathbb{Q}(x) \not\equiv s_\mathbb{Q}(y) \pmod{2}\) and \(N_\mathbb{Q}(x)D_\mathbb{Q}(y) < N_\mathbb{Q}(y)D_\mathbb{Q}(x)\) then expression \eqref{eq:add_rat} can be written as:
    \begin{equation}
        (-1)^{s_\mathbb{Q}(y)}\frac{N_\mathbb{Q}(y)D_\mathbb{Q}(x) \dotdiv N_\mathbb{Q}(x)D_\mathbb{Q}(y)}{D_\mathbb{Q}(x)D_\mathbb{Q}(y)}
    \end{equation}
\end{enumerate}
The congruence modulo 2 can be expressed with the equivalences:
\begin{gather}
    s_\mathbb{Q}(x) \equiv s_\mathbb{Q}(y) \pmod{2} \Leftrightarrow\even(s_\mathbb{Q}(x) + s_\mathbb{Q}(y)) \\
    s_\mathbb{Q}(x) \not\equiv s_\mathbb{Q}(y) \pmod{2} \Leftrightarrow\odd(s_\mathbb{Q}(x) + s_\mathbb{Q}(y))
\end{gather}
where the primitive recursive relations \(\even, \odd \subseteq \mathbb{N}\) are defined via their characteristic functions by:
\begin{gather}
    \begin{cases}
        \chi_{\even}(0) = 1 \\
        \chi_{\even}(n+1) = 1 \dotdiv \chi_{\even}(n)
    \end{cases} \\
    \chi_{\odd}(n) = 1 \dotdiv \chi_{\even}(n)
\end{gather}

Using these, the three cases are described by the primitive recursive relations:
\begin{gather}
    R_1(x,y) \Leftrightarrow \even(s_\mathbb{Q}(x) + s_\mathbb{Q}(y)) \\
    \begin{aligned}
    R_2(x,y) \Leftrightarrow &\odd(s_\mathbb{Q}(x) + s_\mathbb{Q}(y))\\&\land N_\mathbb{Q}(x)D_\mathbb{Q}(y) \geq N_\mathbb{Q}(y)D_\mathbb{Q}(x)        
    \end{aligned} \\
    \begin{aligned}
    R_3(x,y) \Leftrightarrow &\odd(s_\mathbb{Q}(x) + s_\mathbb{Q}(y))\\&\land N_\mathbb{Q}(x)D_\mathbb{Q}(y) < N_\mathbb{Q}(y)D_\mathbb{Q}(x)        
    \end{aligned}
\end{gather}

Therefore, the sign, numerator and denominator of \(\rat(x) + \rat(y)\) are given by the primitive recursive functions:
\begin{gather}
    s_+(x,y) = \begin{cases}
        s_\mathbb{Q}(x), &\text{if }  R_1(x,y)\\
        s_\mathbb{Q}(x), &\text{if }  R_2(x,y)\\
        s_\mathbb{Q}(y), &\text{if }  R_3(x,y)
    \end{cases} \\
    N_+(x,y) = \begin{cases}
    N_\mathbb{Q}(x)D_\mathbb{Q}(y) + N_\mathbb{Q}(y)D_\mathbb{Q}(x), &\text{if } R_1(x,y)\\
    N_\mathbb{Q}(x)D_\mathbb{Q}(y) \dotdiv N_\mathbb{Q}(y)D_\mathbb{Q}(x), &\text{if } R_2(x,y)\\
    N_\mathbb{Q}(y)D_\mathbb{Q}(x) \dotdiv N_\mathbb{Q}(x)D_\mathbb{Q}(y), &\text{if } R_3(x,y)
    \end{cases}
     \\
    D_+(x,y) = D_\mathbb{Q}(x)D_\mathbb{Q}(y)
\end{gather}

Finally, \(+_\mathbb{Q}\) can be defined by:
\begin{equation}
    x +_\mathbb{Q} y = \langle s_+(x,y), N_+(x,y), D_+(x,y) \rangle
\end{equation}

As for the operation \(/_\mathbb{Q}\), if \(\rat(y) \neq 0\), we have:
\begin{align}
    \rat(x)/\rat(y) = (-1)^{s_\mathbb{Q}(x) + s_\mathbb{Q}(y)}\frac{N_\mathbb{Q}(x)D_\mathbb{Q}(y)}{N_\mathbb{Q}(y)D_\mathbb{Q}(x)}
\end{align}

We define the primitive recursive functions:
\begin{gather}
    s_/(x,y) = s_\mathbb{Q}(x) + s_\mathbb{Q}(y) \\
    N_/(x,y) = N_\mathbb{Q}(x)D_\mathbb{Q}(y) \\
    D_/(x,y) = N_\mathbb{Q}(y)D_\mathbb{Q}(x)
\end{gather}

The function \(/_\mathbb{Q}\) can then be defined by:
\begin{equation}
    x /_\mathbb{Q} y = \begin{cases}
        \langle s_/(x,y), N_/(x,y), D_/(x,y) \rangle , &\text{if } N_\mathbb{Q}(y) \neq 0 \\
        0, &\text{otherwise}
    \end{cases}
\end{equation}

For the relations \(=_\mathbb{Q}, \leq_\mathbb{Q}, <_\mathbb{Q}\) we have:
\begin{gather}
\begin{aligned}
    &\rat(x) = \rat(y) \Leftrightarrow\\
    & \left( \even(s_\mathbb{Q}(x) + (s_\mathbb{Q}(y)) \land  N_\mathbb{Q}(x) D_\mathbb{Q}(y) = N_\mathbb{Q}(y) D_\mathbb{Q}(x) \right) \\
        &\lor \left( N_\mathbb{Q}(x) = 0 \land N_\mathbb{Q}(y) = 0 \right)
\end{aligned} \\
\begin{aligned}
    \rat(x) < &\rat(y) \Leftrightarrow\\
    & ( \even(s_\mathbb{Q}(x)) \land \even(s_\mathbb{Q}(y)) \\
    &\quad\quad\quad\quad \land N_\mathbb{Q}(x) D_\mathbb{Q}(y) < N_\mathbb{Q}(y) D_\mathbb{Q}(x) ) \\
    &\lor ( \odd(s_\mathbb{Q}(x)) \land \odd(s_\mathbb{Q}(y)) \\
    &\quad\quad\quad\quad \land N_\mathbb{Q}(x) D_\mathbb{Q}(y) > N_\mathbb{Q}(y) D_\mathbb{Q}(x) ) \\
    &\lor ( \odd(s_\mathbb{Q}(x)) \land \even(s_\mathbb{Q}(y)) \\
    &\quad\quad\quad\quad \land \neg (N_\mathbb{Q}(x) = 0 \land N_\mathbb{Q}(y) = 0) )
\end{aligned} \\
\rat(x) \leq \rat(y) \Leftrightarrow (\rat(x) = \rat(y)) \lor (\rat(x) < \rat(y))
\end{gather}

Therefore, the relations:
\begin{gather}
    \{(x,y) \in \mathbb{N}^2 \mid \rat(x) = \rat(y) \} \\
    \{(x,y) \in \mathbb{N}^2 \mid \rat(x) \leq \rat(y) \} \\
    \{(x,y) \in \mathbb{N}^2 \mid \rat(x) < \rat(y) \}
\end{gather}
are primitive recursive. 
By expressions \eqref{eq:rat=}, \eqref{eq:ratleq}, \eqref{eq:rat<} we conclude that the relations \(=_\mathbb{Q}\), \(\leq_\mathbb{Q}\), \(<_\mathbb{Q}\) are primitive recursive.
\end{proof}

Finding the maximum of two rational numbers is also primitive recursive.

\begin{lemma}
    There exists a primitive recursive function \(\max_{\mathbb{Q}} : \mathbb{N}^2 \rightarrow \mathbb{N}\) such that for all \(x,y \in \mathbb{N}\) with \(\isRat(x)\) and \(\isRat(y)\), the number \(\max_{\mathbb{Q}}(x,y)\) is a code of the rational number \(\max\{\rat(x), \rat(y)\}\).
\end{lemma}

\begin{proof}
    \(\max_\mathbb{Q}(x,y)\) can be defined by:
    \begin{equation}
        \max_{\mathbb{Q}}(x,y) = \begin{cases}
            x, &\text{if } y \leq_{\mathbb{Q}} x \\
            y, &\text{otherwise}
        \end{cases}
    \end{equation}
    Since \(\leq_\mathbb{Q}\) is primitive recursive, \(\max_{\mathbb{Q}}\) is also primitive recursive.
\end{proof}

\section{Computability Framework}\label{sec:computability_framework}

We now introduce the notion of \emph{computable real numbers} and extend the constructions of Subsection \ref{subsec:rationals} to this domain. The class of computable real numbers was first defined by Turing in the same work in which he introduced Turing machines \cite{turing1936computable}. Since then, computable analysis has developed into a rich field in theoretical computer science, and many equivalent definitions have been proposed. Here, we present one based on \(\mu\)-recursive functions. For a detailed treatment, we refer the reader to \cite{weihrauch2000computable}.

\begin{definition}[Computable Real Number]\label{def:computable}
    A number \(\alpha \in \mathbb{R}\) is called computable iff there exist recursive functions \(s\), \(N\), \(D : \mathbb{N} \rightarrow \mathbb{N}\) with \(D(n) \neq 0\) for all \(n \in \mathbb{N}\) that satisfy:
    \begin{equation}
        \left| \alpha - (-1)^{s(n)} \cdot \frac{N(n)}{D(n)} \right| < \frac{1}{2^n}, \quad \forall n \in \mathbb{N}
    \end{equation}
    The set of all computable real numbers is denoted by \(\mathbb{R}_c\).
\end{definition}

Using Definition \ref{def:computable} and Theorem \ref{theorem:NF}, we can define an encoding of the computable real numbers into the natural numbers, analogous to the encoding of Definition \ref{def:rationals}. This encoding then allows us to extend the constructions of Subsection \ref{sec:recursion_framework} and to formalize recursion over computable real numbers.

\begin{definition}[Encoding of Computables]
    Let \(\alpha \in \mathbb{R}_c\) with corresponding recursive functions \(s\), \(N\), \(D : \mathbb{N} \rightarrow \mathbb{N}\) satisfying \(D(n) \neq 0\) for all \(n \in \mathbb{N}\) and:
    \begin{equation}
        \left| \alpha - (-1)^{s(n)} \cdot \frac{N(n)}{D(n)} \right| < \frac{1}{2^n}, \quad \forall n \in \mathbb{N}
    \end{equation}
    We call the recursive function \(f : \mathbb{N} \rightarrow \mathbb{N}\) defined by:
    \begin{equation}
        f(n) = \langle s(n), N(n), D(n) \rangle
    \end{equation}
    a recursive rational approximation of \(\alpha\) and we refer to each code of \(f\) as a code of \(\alpha\).

    We define the relation \(\isCom \subseteq \mathbb{N}\) by:
    \begin{equation}
        \isCom(u) \Leftrightarrow u \text{ is a code of some computable real } \alpha
    \end{equation}
    and the function \(\com : \isCom \rightarrow \mathbb{R}_c\) that satisfies \(\com(u) = \alpha\) whenever \(u\) is a code of \(\alpha\) by:
    \begin{equation}
        \com(u) = \lim_{n \rightarrow \infty}\rat(\varphi^1(u, n))
    \end{equation}
\end{definition}

We can now formulate the recursiveness of addition and multiplication over the computable real numbers, as stated in Lemma \ref{lemma:computable_operations}.

\begin{lemma}\label{lemma:computable_operations}
    There exist primitive recursive functions \(+_{\mathbb{R}_c}\), \(\cdot_{\mathbb{R}_c} : \mathbb{N}^2 \rightarrow \mathbb{N}\) such that for all \(x,y \in \mathbb{N}\) with \(\isCom(x)\) and \(\isCom(y)\) the following conditions hold:
    \begin{gather}
        \com(x +_{\mathbb{R}_c}y) = \com(x) + \com(y) \\
        \com(x \cdot_{\mathbb{R}_c}y) = \com(x) \cdot \com(y)
    \end{gather}
\end{lemma}

\begin{proof}
    Let \(x,y \in \mathbb{N}\) with \(\isCom(x)\) and \(\isCom(y)\). Set \(\alpha = \com(x) \in \mathbb{R}_c\) and \(\beta = \com(y) \in \mathbb{R}_c\). Denote by \(f_\alpha, f_\beta\) the recursive rational approximations of \(\alpha\) and \(\beta\) respectively:
    \begin{gather}
        f_\alpha(n) = \varphi^1(x,n) \\
        f_\beta(n) = \varphi^1(y,n)
    \end{gather}
    Let \(\epsilon_\alpha(n) = \rat(f_\alpha(n)) - \alpha\) and \(\epsilon_\beta(n) = \rat(f_\beta(n)) - \beta\) denote the approximation errors. By definition we have for all \(n \in \mathbb{N}\):
    \begin{gather}
        \left| \epsilon_\alpha(n) \right|, \ \left| \epsilon_\beta(n) \right| < \frac{1}{2^n}
    \end{gather}
    
    For the addition, note that the function \(\rat(f_\alpha(n+1)) + \rat(f_\beta(n+1))\) satisfies:
    \begin{align}
        | \rat(f_\alpha(n+1)) + \rat(&f_\beta(n+1)) - (\alpha + \beta) | =\\
        & \left| \epsilon_\alpha(n+1) + \epsilon_\beta(n+1) \right| < \frac{1}{2^n}
    \end{align}
    Therefore, the function:
    \begin{equation}\label{eq:com_add}
        f_+(n) = f_\alpha(n+1) +_\mathbb{Q} f_\beta(n+1)
    \end{equation}
    is a recursive rational approximation of \(\alpha + \beta\). By Lemma \ref{lemma:effectiveness} and equation \eqref{eq:com_add}, we can construct a primitive recursive function \(+_{\mathbb{R}_c}\) with \(x +_{\mathbb{R}_c}y\) being a code of the function \(f_+\). Therefore, \(x +_{\mathbb{R}_c}y\) is a code of the computable real number \(\alpha + \beta\).

    For the multiplication we have:
    \begin{align}
        \rat(f_\alpha(k)) &\cdot\rat(f_\beta(k)) =\\&\alpha\beta + \alpha\epsilon_\beta(k) + \beta\epsilon_\alpha(k) + \epsilon_\alpha(k)\epsilon_\beta(k)
    \end{align}
    Therefore, we have:
    \begin{align}
        |\rat(f_\alpha(k))& \cdot\rat(f_\beta(k)) - \alpha\beta|\\
        &\leq |\alpha\epsilon_\beta(k)| + |\beta\epsilon_\alpha(k)| + |\epsilon_\alpha(k)\epsilon_\beta(k)| \\
        &<\frac{|\alpha| + |\beta|}{k} + \frac{1}{2^{k+1}}\label{eq:mult_ineq}
    \end{align}

    We know that \(\rat(f_\alpha(0)) - 1 < \alpha < \rat(f_\alpha(0))+1\). Hence:
    \begin{align}
        |\alpha| &< \max\{|\rat(f_\alpha(0)) - 1|, |\rat(f_\alpha(0)) + 1|\} \\
        &\leq |\rat(f_\alpha(0))| + 1 \\
        &\leq  N_{\mathbb{Q}}(\varphi^1(x,0)) + 1
    \end{align}




    Define the recursive function:
    \begin{equation}
            M(x) = N_{\mathbb{Q}}(\varphi^1(x,0)) + 1
    \end{equation}
    By construction, \(M\) satisfies \(M(x) > |\alpha|\), \(M(y) > |\beta|\) and \(M(n) \neq 0\). By setting \(k = K(n,x,y) = (M(x)+M(y))2^{n+1} > n\) in the inequality \eqref{eq:mult_ineq} we achieve:
    \begin{align}
        |\rat(f_\alpha(&K(n,x,y))) \cdot\rat(f_\beta(K(n,x,y))) - \alpha\beta| \\
        & <\frac{|\alpha| + |\beta|}{(|\alpha| + |\beta|)2^{n+1}} + \frac{1}{2^{K(n,x,y)+1}} < \frac{1}{2^n} 
    \end{align}

    Therefore, the function:
    \begin{equation}
        f_\cdot(x,y,n) = f_\alpha(K(n,x,y)) \cdot_\mathbb{Q} f_\beta(K(n,x,y))
    \end{equation}
    is a recursive rational approximation of \(\alpha \cdot \beta\), when viewed as a function of only \(n\). By Lemma \ref{lemma:effectiveness} we can construct a primitive recursive function \(F_\cdot\), with \(F_\cdot(x,y)\) being a code of \(f_\cdot(x,y,n)\). The function \(\cdot_{\mathbb{R}_c}\) can then be defined by:
    \begin{equation}
        x \cdot_{\mathbb{R}_c} y = S^2_1(F_\cdot(x,y), x,y)
    \end{equation}
\end{proof}

Subtraction and division over \(\mathbb{R}_c\) can similarly be shown to be recursive, although we will not require these operations in our analysis. Notably, the relations of equality and ordering of computable real numbers are not recursive.

\begin{lemma}\label{lemma:non_computable_relations}
The relations \(=_{\mathbb{R}_c}\), \(<_{\mathbb{R}_c} \subseteq \mathbb{N}^2\) defined by:
\begin{gather}
    x =_{\mathbb{R}_c} y \Leftrightarrow \isCom(x) \land \isCom(y) \land \com(x) = \com(y) \\
    x <_{\mathbb{R}_c} y \Leftrightarrow \isCom(x) \land \isCom(y) \land \com(x) < \com(y)
\end{gather}
are not recursive.
\end{lemma}

\begin{proof}
We will reduce the halting problem to both \(=_{\mathbb{R}_c}\) and \(<_{\mathbb{R}_c}\). Since the halting problem is famously non recursive, we conclude that \(=_{\mathbb{R}_c}\) and \(<_{\mathbb{R}_c}\) are also non recursive.

The halting problem can be formulated, using the universal recursive functions, as determining for arbitrary \(e,x \in \mathbb{N}\) whether \(\varphi^1(e,x) \downarrow\) or not. The fact that it is not solvable algorithmically is expressed as the non recursiveness of the relation \(H \subseteq \mathbb{N}^2\) defined by:
\begin{equation}
    H(e, x) \Leftrightarrow \varphi^1(e,x) \downarrow
\end{equation}

We proceed to defining the relation \(H(e,x)\) in terms of \(=_{\mathbb{R}_c}\) and \(<_{\mathbb{R}_c}\). Let \(e,x \in \mathbb{N}\) be arbitrary natural numbers. Define the function:
\begin{equation}
    f(e,x,k) = \begin{cases}
        1, &\text{if } \exists y \leq k:( T_1(y,e,x)) \\
        0, &\text{otherwise}
    \end{cases}
\end{equation}
\(f\) is primitive recursive, since the relation \(T_1\) is primitive recursive by Theorem \ref{theorem:NF}. Furthermore, \(H(e,x)\) is true iff \(f(e,x,k) = 1\) for some \(k \in \mathbb{N}\). Define a second function \(f_M : \mathbb{N}^3 \rightarrow \mathbb{N}\) by the primitive recursion:
\begin{equation}
\begin{cases}
    f_M(0,e,x) = f(e,x,0)\cdot\langle0,1,1\rangle \\
    \begin{aligned}
    f_M(n+1,e,x) = &f_M(n,e,x) +_\mathbb{Q} \\
    &f(e,x,n+1)\cdot\langle0,1,2^{n+1}\rangle
    \end{aligned}
\end{cases}
\end{equation}
The definition od \(f_M\) is such that the number \(f_M(n,e,x)\) is a code of the rational number:
\begin{equation}
    q_{n,e,x} = \sum_{k=0}^nf(e,x,k) \cdot \frac{1}{2^k}
\end{equation}
Therefore, the function \(f_M'(n) = f_M(n,e,x)\) is a recursive rational approximation of some computable real number \(\beta\), which is zero iff \(f(e,x,n) = 0\) for all \(n \in \mathbb{N}\) and positive otherwise. This means that:
\begin{equation}
    \begin{cases}
        0 = \beta \Leftrightarrow \neg H(e,x) \\
        0 < \beta \Leftrightarrow H(e,x)
    \end{cases}
\end{equation}

Let \(c_0\) be a code of the computable real number \(0\). Let \(c_M\) be a code of the recursive function \(\hat{f}(e,x,n) = f_M(n,e,x)\). A code of the function \(f_M'(n) = \hat{f}(e,x,n)\), and therefore a code of \(\beta\) is given by \(S^2_1(c_M,e,x)\). Hence, we arrive at the equivalences:
\begin{gather}
    H(e,x) \Leftrightarrow \neg(c_0 =_{\mathbb{R}_c} S^2_1(c_M, e,x)) \label{eq:com_equality}  \\
    H(e,x) \Leftrightarrow c_0 <_{\mathbb{R}_c} S^2_1(c_M, e,x) \label{eq:com_ordering}
\end{gather}

From the above expressions we see that if either of the relations \(=_{\mathbb{R}_c}, <_{\mathbb{R}_c}\) is computable, then \(H\) is also computable, which is false. Therefore, by contradiction, \(=_{\mathbb{R}_c}\) and \(<_{\mathbb{R}_c}\) are not computable.

\end{proof}

In fact, even the comparison of computable real numbers with rationals is not recursive, as stated in Lemma \ref{lemma:impossibility}.

\begin{lemma}\label{lemma:impossibility}
    The relation \(<_{\mathbb{R}_c,\mathbb{Q}} \subseteq \mathbb{N}^2\) defined by:
    \begin{equation}
        x <_{\mathbb{R}_c,\mathbb{Q}} y \Leftrightarrow \com(x) < \rat(y)
    \end{equation}
    is not recursive.
\end{lemma}

\begin{proof}
    We use the exact same logic as with the proof of Lemma \ref{lemma:non_computable_relations}. Let \(c_0\) be a code of the rational number 0. Let \(f(e,x,k)\) be the same as in the proof of Lemma \ref{lemma:non_computable_relations}. Define \(f_M\) by the primitive recursion:
    \begin{equation}
    \begin{cases}
        f_M(0,e,x) = f(e,x,0)\cdot\langle 1, 1, 1 \rangle \\
    \begin{aligned}
    f_M(n+1,e,x) = &f_M(n,e,x) +_\mathbb{Q} \\
    &f(e,x,n+1)\cdot\langle1,1,2^{n+1}\rangle
    \end{aligned}
    \end{cases}
    \end{equation}
so that \(f_M(n,e,x)\) is a code of the rational number:
\begin{equation}
    q_{n,e,x} = \sum_{k=0}^nf(e,x,k)\cdot \frac{-1}{2^n}
\end{equation}
Let \(c_M\) be a code of the computable function \(\hat{f}(e,x,n) = f_M(n,e,x)\). Then we have the equivalence:
\begin{equation}
    H(e,x) \Leftrightarrow S^2_1(c_M,e,x) <_{\mathbb{R}_c, \mathbb{Q}} c_0
\end{equation}
from which we conclude that \(<_{\mathbb{R}_c, \mathbb{Q}}\) is not computable.
\end{proof}

Although the order relation on computable real numbers is not recursive, it is still possible to compute the maximum of two computable reals in a recursive way.

\begin{lemma}
    There exists a primitive recursive function \(\max_{\mathbb{R}_c} : \mathbb{N}^2 \rightarrow \mathbb{N}\) such that for all \(x,y \in \mathbb{N}\) with \(\isCom(x)\) and \(\isCom(y)\), the number \(\max_{\mathbb{R}_c}(x,y)\) is a code of the computable real number \(\max\{\com(x), \com(y)\}\).
\end{lemma}

\begin{proof}
    Let \(x,y \in \mathbb{N}\) satisfy \(\isCom(x)\) and \(\isCom(y)\). Set \(\alpha = \com(x)\), \(\beta = \com(y)\), \(f_\alpha(n) = \varphi^1(x,n)\) and \(f_\beta(n) = \varphi^1(y,n)\). By definition, we have for all \(n \in \mathbb{N}\):
    \begin{gather}
        \alpha-\frac{1}{2^n} < \rat(f_\alpha(n)) < \alpha + \frac{1}{2^n} \\
        \beta-\frac{1}{2^n} < \rat(f_\beta(n)) < \beta + \frac{1}{2^n}
    \end{gather}
It follows that:
\begin{align}
    &\max\left\{\alpha-\frac{1}{2^n}, \beta - \frac{1}{2^n}\right\} < \max\left\{\rat(f_\alpha(n)),  \rat(f_\beta(n)) \right\} \\
    &\quad\quad\quad\quad\quad\quad\quad\quad\quad\quad< \max\left\{\alpha+\frac{1}{2^n}, \beta + \frac{1}{2^n}\right\} \\
    &\Leftrightarrow \max\{\alpha, \beta\} -  \frac{1}{2^n} < \rat\left(\max_{\mathbb{Q}}\{f_\alpha(n), f_\beta(n)\}\right) \\
    & \quad\quad\quad\quad\quad\quad\quad\quad\quad < \max\{\alpha, \beta\} + \frac{1}{2^n}
\end{align}
Therefore, the function:
\begin{equation}\label{eq:max}
    f_{\max}(n) = \max_{\mathbb{Q}}(f_\alpha(n), f_\beta(n))
\end{equation}
is a recursive rational approximation of \(\max\{\alpha, \beta\}\)
The existence of \(\max_{\mathbb{R}_c}\) follows from equation \eqref{eq:max} and Lemma \ref{lemma:effectiveness}.
\end{proof}

\section{Problem Formulation and Solution}\label{sec:solution}

We now proceed with a precise formulation of the problem. Our goal is to determine whether the task of constructing capacity-achieving codes is computationally solvable, in the sense of Question \ref{q:problem}, which is a more formal version of Question \ref{q:problem_informal}:

\begin{question}\label{q:problem}
Does there exist a Turing machine \(M = (Q, \Gamma, \square, \Sigma, \delta, q_0, F)\) such that, when given as input a DMC \(p_{Y \mid X}\), a rate \(R < C(p_{Y \mid X})\) and an error tolerance \(\epsilon > 0\), all appropriately encoded over the input alphabet \(\Sigma\), it outputs a description over the tape alphabet \(\Gamma\) of a block code \(\mathcal{C}\) with rate at least \(R\) and with maximum block error probability \(\lambda_{\max} < \epsilon\)?
\end{question}

Note that the channel probabilities \(p_{Y \mid X}(y \mid x)\), as well as the numbers \(R\) and \(\epsilon\) are generally real numbers. However, there is no injective encoding of the set \(\mathbb{R}\) over any finite alphabet \(\Sigma\), since \(\left| \mathbb{R} \right| = 2^{\aleph_0} > \aleph_0 = \left| \Sigma^* \right|\). The most general subset of \(\mathbb{R}\) that can be encoded over \(\Sigma\) and for which we can perform computations using a Turing machine, is the set of computable real numbers \(\mathbb{R}_c\). To this end, we will prove that there exists a Turing machine \(M\) that satisfies the conditions of Question \ref{q:problem}, for any DMC with computable probabilities \(p_{Y \mid X}(y \mid x)\) and for any rational values of \(R\) and \(\epsilon\).

To do this, we will construct a \(\mu\)-recursive function \(\FindCode : \mathbb{N}^3 \rightharpoonup \mathbb{N}\) that takes as inputs appropriate encodings of \(p_{Y \mid X}\), \(R\) and \(\epsilon\) over the natural numbers, and outputs an encoding over \(\mathbb{N}\) of a block code that satisfies the required conditions. Since, by Theorem \ref{theorem:NF}, Turing machines and \(\mu\)-recursive functions are computationally equivalent, we conclude that the problem can be solved by a Turing machine.

We will then extend the result to cover the case where \(\epsilon\) is an arbitrary computable real number, rather than just a rational, and we will also discuss how to generalize the result to allow \(R\) to be any computable real number as well.

\subsection{Formulation with Pseudocode}

We will first describe an algorithm that solves the problem using pseudocode and then construct the function \(\FindCode\) based on this pseudocode. A naive first formulation is given by Algorithm \ref{alg:algorithm_naive}.

\begin{algorithm}
    \caption{Naive approach}\label{alg:algorithm_naive}
    \hspace*{\algorithmicindent} \textbf{Input:} \(p_{Y \mid X}, R, \epsilon\) \\
    \hspace*{\algorithmicindent} \textbf{Output:} \(\mathcal{C}\) 
    \begin{algorithmic}[1]
    \State \(n \gets 1\)
    \While{\textbf{True}}
    \For{\(\mathcal{C} \in \textsc{Codes}\left(p_{Y \mid X},\left\lceil 2^{nR} \right\rceil, n\right)\)}
        \State \(\lambda \gets \lambda_{\max}(\mathcal{C}, p_{Y \mid X})\)
        \If{\(\lambda < \epsilon\)} 
        \State \Return \(\mathcal{C}\)
        \EndIf
    \EndFor
    \State \(n \gets n+1\)
    \EndWhile
    \end{algorithmic}
\end{algorithm}

Algorithm \ref{alg:algorithm_naive} can be summarized as follows:
\begin{enumerate}
    \item Initialize \(n \gets 1\).
    \item Generate the list \(\textsc{Codes}\left(p_{Y \mid X}, \left\lceil 2^{nR} \right\rceil, n\right)\) consisting of all \(\left( \left\lceil 2^{nR} \right\rceil, n \right)\) block codes for the channel \(p_{Y \mid X}\). This is possible because, for given \(m = \left\lceil 2^{nR} \right\rceil\) and \(n\), the number of such block codes is finite and equal to \(m^{|X|^n} \cdot |Y|^{m \cdot n}\).  {Note that the rate of an \((m,n)\) block code is defined as \(\frac{\log_2 m}{n}\), so \(\left( \left\lceil 2^{nR} \right\rceil, n \right)\) block codes have by definition rate at least \(R\).}
    \item For each code \(\mathcal{C} \in \textsc{Codes}\left(p_{Y \mid X}, \left\lceil 2^{nR} \right\rceil, n\right)\), compute the  {maximum block error probability of \(\mathcal{C}\) under \(p_{Y \mid X}\), denoted by \(\lambda = \lambda_{\max}(\mathcal{C}, p_{Y \mid X})\)}. This step is feasible, since the operations required to compute \(\lambda\) are recursive. After computing \(\lambda\), check whether \(\lambda < \epsilon\). If the condition holds, return the code \(\mathcal{C}\). Otherwise, proceed to the next code.
     \item If none of the codes \(\mathcal{C} \in \textsc{Codes}\left(p_{Y \mid X}, \left\lceil 2^{nR} \right\rceil, n\right)\) satisfies the condition, increment \(n\) by one and repeat the process for the new codeword length.
\end{enumerate}

This algorithm proceeds in a exhaustive search fashion by enumerating all possible block codes for increasing codeword lengths \(n\), starting from \(n = 1\), until a code satisfying the error constraint is found. When \(R < C(p_{Y \mid X})\), Shannon's channel coding theorem guarantees the existence of such codes for all \(n \geq n_0\), for some threshold \(n_0\). Therefore, the algorithm is guaranteed to terminate with a valid block code \(\mathcal{C}\) such that \(\lambda_{\max}(\mathcal{C}, p_{Y \mid X}) < \epsilon\).

The issue with Algorithm \ref{alg:algorithm_naive} is that \(\lambda_{\max}\) is, in general, a computable real number, and the truth of the expression \(\lambda_{\max} < \epsilon\) cannot be recursively decided, as stated in Lemma \ref{lemma:impossibility}. For this reason, we modify the algorithm based on the observations of Lemma \ref{lemma:observations}.

\begin{lemma}\label{lemma:observations}
The following are true:
\begin{enumerate}
    \item There exists a recursive function \(\BLB : \mathbb{N} \rightharpoonup \mathbb{N}\) (acronym for Binary Lower Bound) such that, if \(c_\epsilon \in \mathbb{N}\) with \(\rat(c_\epsilon) = \epsilon > 0\), \(\epsilon \in \mathbb{Q}\) then \(\BLB(c_\epsilon) \downarrow\) and \(b = \BLB(c_\epsilon) \in \mathbb{N}\) satisfies \(2^{-b} < \epsilon\).
    \item For a DMC \(p_{Y \mid X}\) with \(C(p_{Y \mid X}) \neq 0\) and a positive rate \(R < C(p_{Y \mid X})\), there exists a \(\left( \left\lceil 2^{nR} \right\rceil , n\right)\) block code \(\mathcal{C}\) with \(\lambda_{\max}(\mathcal{C}, p_{Y \mid X}) < 2^{-b-2}\), for any \(b \in \mathbb{N}\).
    \item For \(c_\lambda \in \mathbb{N}\) with \(\com(c_\lambda) = \lambda < 2^{-b-2}\), \(\lambda \in \mathbb{R}_c\) we have \(\rat\left(\varphi^1\left(c_\lambda, b+2\right)\right) < 2^{-b-1}\).
    \item For \(c_\lambda \in \mathbb{N}\) with \(\com(c_\lambda) = \lambda \in \mathbb{R}_c\) and \(\rat\left(\varphi^1\left(c_\lambda, b+2\right)\right) < 2^{-b-1}\) we have \(\lambda < 2^{-b}\).
\end{enumerate}
\end{lemma}

\begin{proof}
    \begin{enumerate}
        \item Since \(\epsilon > 0\), it holds that \(2^{-n} < \epsilon\) for all natural numbers \(n > -\log(\epsilon)\). Therefore, the minimization:
        \begin{equation}
            \BLB(c_\epsilon) = \mu i :( \langle 0, 1, 2^i \rangle <_\mathbb{Q} c_\epsilon)
        \end{equation}
        converges and it returns a number \(b\) with \(2^{-b} < \epsilon\).

        \item It follows immediately from the channel coding theorem.

        \item We have:
        \begin{align}
            \rat\left( \varphi^1(c_\lambda, b+2)\right) &< \lambda + \frac{1}{2^{b+2}} <\frac{1}{2^{b+1}} 
        \end{align}

        \item We have:
        \begin{align}
            \lambda &< \rat\left( \varphi^1(c_\lambda, b+2) \right) + \frac{1}{2^{b+2}} \\
            &< \frac{1}{2^{b+1}} + \frac{1}{2^{b+2}} <\frac{1}{2^b}
        \end{align}
    \end{enumerate}
\end{proof}

With Lemma \ref{lemma:observations} in mind we proceed with the following reasoning:
\begin{enumerate}
    \item Calculate \(b = \BLB(c_\epsilon)\).
    \item There exists a block code \(\mathcal{C}\) such that \(\lambda = \lambda_{\max}(\mathcal{C}, p_{Y \mid X}) < 2^{-b-2}\). For any such code, it holds that \(\rat\left(\varphi^1\left(c_\lambda, b+2\right)\right) < 2^{-b-1}\), where \(c_\lambda\) is any code of \(\lambda \in \mathbb{R}_c\).
    \item We search through all possible block codes for \(p_{Y \mid X}\) and return the first one that satisfies \(\rat\left(\varphi^1\left(c_\lambda, b+2\right)\right) < 2^{-b-1}\). By the previous point, we will eventually find such a code.
    \item The returned code \(\mathcal{C}\) also satisfies \(\lambda_{\max}(\mathcal{C}, p_{Y \mid X}) < 2^{-b} < \epsilon\).
\end{enumerate}

The modified pseudocode is presented in Algorithm \ref{alg:algorithm_working}.

\begin{algorithm}
    \caption{Working approach}\label{alg:algorithm_working}
    \hspace*{\algorithmicindent} \textbf{Input:} \(p_{Y \mid X}, R, \epsilon\) \\
    \hspace*{\algorithmicindent} \textbf{Output:} \(\mathcal{C}\) 
    \begin{algorithmic}[1]
    \State \(b \gets \BLB(\epsilon)\)
    \State \(n \gets 1\)
    \While{\textbf{True}}
    \For{\(\mathcal{C} \in \textsc{Codes}\left(p_{Y \mid X},\left\lceil 2^{nR} \right\rceil, n\right)\)}
            \State \(c_{\lambda} \gets \text{ a code of } \lambda_{\max}(\mathcal{C}, p_{Y \mid X})\)
            \If{\(\rat\left(\varphi^1\left(c_{\lambda}, b+2\right)\right) < 2^{-b-1}\)} 
                \State \Return \(\mathcal{C}\)
            \EndIf
    \EndFor
    \State \(n \gets n+1\)
    \EndWhile
    \end{algorithmic}
\end{algorithm}

Note that the condition \(\rat\left(\varphi^1\left(c_{\lambda}, k+2\right)\right) < 2^{-k-1}\) is sufficient for ensuring that the code satisfies \(\lambda_{\max}(\mathcal{C}, p_{Y \mid X}) < 2^{-b-1}\), but it is not necessary. This means that the code \(\mathcal{C}\) that is eventually returned, although guaranteed to meet the error constraint, may not be the first code in the enumeration that actually satisfies  \(\lambda_{\max}(\mathcal{C}, p_{Y \mid X}) < 2^{-b-1}\).

In the following analysis, we will define the recursive function \(\FindCode\) based on Algorithm \ref{alg:algorithm_working}. We will gradually construct intermediate recursive functions and relations that solve smaller parts of the problem. In the end, we will combine everything to achieve the full solution. 

Specifically, we will proceed through the following steps:
\begin{enumerate}
    \item Define an encoding of the class of all DMCs over the natural numbers.
    \item Define an encoding of the class of all block codes over the natural numbers.
    \item Construct a function \(\Codes : \mathbb{N}^4 \rightarrow \mathbb{N}\) that takes as input the sizes \(M\) and \(N\) of the input and output alphabets of a DMC \(p_{Y \mid X}\), along with parameters \(m\) and \(n\), and returns an encoding of the list of all  \((m, n)\) block codes for \(p_{Y \mid X}\).
    \item Construct a function \(\Lambda : \mathbb{N}^2 \rightarrow \mathbb{N}\) that takes as input the encodings of a code \(\mathcal{C}\) and a DMC \(p_{Y \mid X}\) and outputs a code of the computable number \(\lambda_{\max}(\mathcal{C}, p_{Y \mid X})\).
    \item Construct a function \(\AchievesError : \mathbb{N}^3 \rightarrow \mathbb{N}\) that takes as input the code \(c = \langle c_0, c_1, \dots c_{k-1} \rangle\) where \(c_i\) are all encodings of block codes, an encoding \(c_H\) of a DMC and a number \(b \in \mathbb{N}\). The function \(\AchievesError\) returns the index \(i\) of the first code \(c_i\) that satisfies \(\rat\left(\varphi^1\left(\Lambda(c_i, c_H), b+2\right)\right) < 2^{-b-1}\), if such a code exists. Otherwise it returns \(k\).
    \item Construct a function \(\MessageNumber : \mathbb{N}^2 \rightarrow \mathbb{N}\), which satisfies \(\MessageNumber(c_R, n) = \left\lceil 2^{nR} \right\rceil\) for \(R = \rat(c_R) > 0\).
    \item Combine the above to define the function \(\FindCode\).
\end{enumerate}

\subsection{Encoding of DMCs}\label{subsec:DMCs}

In this subsection we define an encoding function \(\Code_\mathcal{H}\) from the set of all DMCs with computable probabilities to the set \(\mathbb{N}\) of natural numbers. This is a crucial step, as we aim to represent Algorithm \ref{alg:algorithm_working} by a recursive function, which means that we have to encode all the inputs of the algorithm into natural numbers.

 {We will identify a DMC with input alphabet \(X\) and output alphabet \(Y\) by its corresponding conditional distribution \(p_{Y \mid X}\).} We will further use the notation \([k] = \{1,2, \dots k\}\) for \(k \in \mathbb{N} \backslash \{0\}\) to denote the set of the first \(k\) positive integers.

A DMC \(p_{Y \mid X}\) with \(X = \{x_1, x_2, \dots x_M\}\) and \(Y = \{y_1, y_2, \dots y_N\}\) can be represented by a matrix \(H \in \mathbb{R}^{M \times N}\), with \(H_{ij} = p_{Y \mid X}(y_j \mid x_i)\). We denote by \(\mathcal{H}\) the set of all matrices that represent DMCs with computable probabilities, as:
\begin{equation}
    \mathcal{H} = \bigcup_{\substack{M,N \in \mathbb{N} \\ M,N \neq 0}} \left\{ H \in \mathbb{I}_c^{M \times N} \middle| \sum_{j=1}^N H_{ij} = 1, \text{ for all } 1 \leq i \leq M  \right\}
\end{equation}
where \(\mathbb{I}_c = [0, 1] \cap \mathbb{R}_c\) denotes the set of computable real numbers in the interval \([0, 1]\).

We define an encoding \(\Code_{\mathcal{H}} : \mathcal{H} \rightarrow \mathbb{N}\). For \(H \in \mathcal{H}\) with dimensions \(M \times N\), choose some \(x_{ij} \in \mathbb{N}\) with \(\com(x_{ij}) = H_{ij}\) for all elements \(H_{ij}\) of \(H\). Then define:
\begin{enumerate}
    \item \(r(H, i) = \langle x_{i1}, x_{i2}, \dots x_{iN} \rangle\) for all \(i \in [M]\)
    \item \(\Code_{\mathcal{H}}(H) = \langle r(H, 1), r(H, 2), \dots r(H, M) \rangle\)
\end{enumerate}

We also define the recursive functions \(\RowNumber\), \(\ColumnNumber : \mathbb{N} \rightarrow \mathbb{N}\) and \(\Element : \mathbb{N}^3 \rightarrow \mathbb{N}\) by:
\begin{gather}
    \RowNumber(c_H) = \lh(c_H) \\
    \ColumnNumber(c_H) = \lh((c_H)_0) \\
    \Element(c_H, i, j) = (c_H)_{i\dotdiv1,j\dotdiv1}
\end{gather}

If \(c_H = \Code_{\mathcal{H}}(H)\) for some \(H \in \mathcal{H}\) with dimensions \(M \times N\), then \(\RowNumber(c_H) = M\), \(\ColumnNumber(c_H) = N\) and \(\Element(c_H, i, j)\) is a code of the computable real number \(H_{ij}\), provided that \(i \in [M]\) and \(j \in [N]\).

\subsection{Encoding of Block Codes}\label{subsec:codes}

Building on the approach of Subsection \ref{subsec:DMCs}, we now define an encoding \(\Code_\mathbf{C}\) that maps the set of all block codes to the set \(\mathbb{N}\) of natural numbers, thereby enabling the use of block codes as inputs to recursive functions. Since a rational error tolerance \(\epsilon\) can also be encoded as a natural number, as discussed in Subsection \ref{subsec:rationals}, this construction allows all inputs to Algorithm \ref{alg:algorithm_working} to be represented using natural numbers.

 {For a \((m,n)\) block code \(\mathcal{C} = (E, D)\), where \(E : \mathcal{M} \rightarrow X^n\) and \(D: Y^n \rightarrow \mathcal{M}\) are the encoding and decoding functions, respectively, and \(\mathcal{M}\) is the set of messages with \(|\mathcal{M}| = m\)}, if \(|X| = M\) and \(|Y| = N\), then we can assume, without loss of generality, that \(\mathcal{M} = [m]\), \(X = [M]\) and \(Y = [N]\).

We denote by \(\mathbf{C}_{M,N,m,n}\) the set of all \((m, n)\) block codes with input and output alphabets of sizes \(M\) and \(N\) respectively, as:
\begin{equation}
    \mathbf{C}_{M,N,m,n} = \left\{ (E, D) \middle| E : [m] \rightarrow [M]^n, D: [N]^n \rightarrow [m] \right\}
\end{equation}

We also denote by \(\mathbf{C}\) the set of all block codes, as:
\begin{equation}
    \mathbf{C} = \bigcup_{\substack{ M,N,m,n \in \mathbb{N} \\ M,N,m,n \neq 0 }} \mathbf{C}_{M,N,m,n}
\end{equation}

We define an encoding \(\Code_{\mathbf{C}} : \mathbf{C} \rightarrow \mathbb{N}\). The definition is done in three steps:
\begin{enumerate}
    \item We define \(\Code_E\), which encodes functions of the form \(E : [m] \rightarrow [M]^n\). If \(E(i) = (x_{i1}, x_{i2}, \dots x_{in})\), define:
    \begin{gather}
        e_i = \langle x_{i1}, x_{i2}, \dots x_{in} \rangle, \text{ for } i \in [m] \\
        \Code_E(E) = \langle e_1, e_2, \dots e_m \rangle
    \end{gather}

    \item Similarly, we define \(\Code_D\), which encodes functions of the form \(D : [N]^n \rightarrow [m]\), as:
    \begin{equation}
        \Code_D(D) = \langle D(\bar{y}_1), D(\bar{y}_2), \dots D(\bar{y}_{N^n}) \rangle
    \end{equation}
    where \(\bar{y}_1, \bar{y}_2, \dots \bar{y}_{N^n}\) is the lexicographic enumeration of all elements in \([N]^n\).

    \item We define \(\Code_\mathbf{C}\) for all \(\mathcal{C} = (E,D) \in \mathbf{C}\) as:
    \begin{equation}
        \Code_\mathbf{C}(\mathcal{C}) = \langle \Code_E(E), \Code_D(D) \rangle
    \end{equation}
\end{enumerate}

We will also need a recursive way to check whether a number \(n \in \mathbb{N}\) is an encoding of some block code \(\mathcal{C} \in \mathbf{C}\). For this reason, we define the primitive recursive relations:

\begin{enumerate}
    \item \(\isCode_E \subseteq \mathbb{N}^5\), with \(\isCode_E(c, M,N,m,n)\) true iff \(c = \Code_E(E)\) for some function \(E : [m] \rightarrow [M]^n\). We have:
    \begin{equation}
        \begin{aligned}
            &\isCode_E(c, M,N,m,n) \Leftrightarrow\\
            &\quad\quad M>0 \land  m >0 \land n>0 \land \seq(c) \land \lh(c)=m  \\
        &\quad\quad\land\forall i \leq m \dotdiv 1: (\seq((c)_i) \land \lh((c)_i) = n \\
        &\quad\quad\land\forall j \leq n \dotdiv 1 : ((c)_{i,j} \geq 1 \land (c)_{i,j} \leq M)     )
        \end{aligned}
    \end{equation}
    
    \item \(\isCode_D \subseteq \mathbb{N}^5\), with \(\isCode_D(c, M,N,m,n)\) true iff \(c = \Code_D(D)\) for some function \(D : [N]^n \rightarrow [m]\). We have:
    \begin{equation}
        \begin{aligned}
            &\isCode_D(c, M,N,m,n) \Leftrightarrow\\
            &\quad\quad N>0 \land  m >0 \land n>0 \land \seq(c) \land \lh(c)=N^n  \\
        &\quad\quad\land\forall i \leq N^n\dotdiv 1: ((c)_{i} \geq 1 \land (c)_{i} \leq m)
        \end{aligned}
    \end{equation}
    
    \item \(\isCode_\mathbf{C} \subseteq \mathbb{N}^5\), with \(\isCode_\mathbf{C}(c, M,N,m,n)\) true iff \(c = \Code_\mathbf{C}(\mathcal{C})\) for some block code \(\mathcal{C} \in \mathbf{C}_{M,N,m,n}\). We have:
    \begin{equation}
        \begin{aligned}
           \isCode_\mathbf{C}&(c, M,N,m,n) \Leftrightarrow \\
           &\seq(c) \land \lh(c)=2 \\
        &\land\isCode_E((c)_0,M,N,m,n)\\ &\land \isCode_D((c)_1,M,N,m,n)
        \end{aligned}
    \end{equation}
\end{enumerate}

Note that, since \(\langle \rangle\) is injective, the encoding \(\isCode_\mathbf{C}\) is a bijection from \(\mathbf{C}\) to \(\isCode_\mathbf{C} \subseteq \mathbb{N}\).

\subsection{Encoding of the Set of all \((m,n)\) Block Codes}

In this subsection we construct a recursive function \(\Codes : \mathbb{N}^4 \rightharpoonup \mathbb{N}\), which takes as input four natural numbers \(M,N,m,n\) and returns an encoding of the list of all \((m,n)\) block codes for a DMC \(p_{Y \mid X}\) with \(|X| = M\) and \(|Y| = n\). This function allows us to represent the set \(\textsc{Codes}\left(p_{Y \mid X}, \left\lceil 2^{nR} \right\rceil, n\right)\) in line 4 of Algorithm \ref{alg:algorithm_working} within the recursion framework.

Recall that for finite sets \(A\) and \(B\), the number of functions \(f : A \rightarrow B\) is \(|B|^{|A|}\). Since a block code \(\mathcal{C} \in \mathbf{C}_{M,N,m,n}\) can be constructed by pairing any two functions \(E : [m] \rightarrow [M]^n\) and \(D : [N]^n \rightarrow [m]\), we see that:
\begin{equation}
    |\mathbf{C}_{M,N,m,n}| = \left(M^n\right)^m \cdot m^{N^n} = M^{m  n} m^{N^n} 
\end{equation}

With this in mind, our goal is to define a primitive recursive function \(\ParCodes : \mathbb{N}^5 \rightarrow \mathbb{N}\), where \(\ParCodes(y, M, N, m, n)\) encodes the list of the first \(y\) block codes in \(\mathbf{C}_{M,N,m,n}\). This will be done inductively, using primitive recursion. Given \(\ParCodes(y, M, N, m, n)\), we construct \(\ParCodes(y+1, M, N, m, n)\) by appending the smallest number \(n\) that encodes a valid block code in \(\mathbf{C}_{M,N,m,n}\) and has not yet appeared in the list.

We define the primitive recursive relation \(\notInside \subseteq \mathbb{N}^2\):
\begin{equation}
\begin{aligned}
    &\notInside(e, c) \Leftrightarrow \\ &\quad\quad\seq(c) \land (\lh(c) > 0 \rightarrow \forall i \leq \lh(c) \dotdiv 1 : ((c)_i \neq e))
\end{aligned}
\end{equation}

\(\notInside(e, c)\) is true iff \(c\) is a code of a sequence that does not contain \(e\). \(\ParCodes\) can be defined by:
\begin{gather}
    \ParCodes(0,M,N,m,n) = \langle \varepsilon \rangle \\
    \begin{aligned}
        \ParCodes&(y+1,M,N,m,n) = \\
        &\app(\ParCodes(y,M,N,m,n), u)
    \end{aligned}
\end{gather}
where:
\begin{equation}
\begin{aligned}
    u = \mu i:(&\isCode_{\mathbf{C}}(i, M,N,m,n) \\
    &\land \notInside(i, \ParCodes(y,M,N,m,n))))    
\end{aligned}
\end{equation}

Since \(\Code_{\mathbf{C}}\) is a bijection from \(\mathbf{C}\) to \(\isCode_{\mathbf{C}}\), there are exactly \(M^{m  n} m^{N^n}\) numbers \(i \in \mathbb{N}\) that satisfy \(\isCode(i,M,N,m,n)\). From this we can conclude that \(\ParCodes(y,M,N,m,n) \downarrow\) when \(M,N,m,n \geq 1\) and \(y \leq M^{m  n} m^{N^n}\). Therefore, \(\Codes\) can be defined by:
\begin{equation}
    \Codes(M,N,m,n) = \ParCodes\left( M^{m n} m^{N^n} ,M,N,m,n\right)
\end{equation}

\subsection{Calculating the Maximum Block Error Probability}

In this subsection we define a recursive function \(\Lambda : \mathbb{N}^2 \rightarrow \mathbb{N}\), which calculates the maximum block error probability of a block code \(\mathcal{C}\) for some DMC \(p_{Y \mid X}\). The inputs to \(\Lambda\) are two natural numbers that are the encodings of \(\mathcal{C}\) and \(p_{Y \mid X}\). The output is a code of the computable real number \(\lambda_{\max}(\mathcal{C}, p_{Y \mid X})\). More precisely, we want to construct \(\Lambda\) so that it satisfies, for all \(\mathcal{C} \in \mathbf{C}\) and \(H \in \mathcal{H}\):
\begin{equation}
    \com(\Lambda(\Code_\mathbf{C}(\mathcal{C}), \Code_\mathcal{H}(H))) = \lambda_{\max}(\mathcal{C}, p_{Y \mid X})
\end{equation}
where \(p_{Y \mid X}\) is the DMC represented by the matrix \(H\).

To do so, we will fist calculate the conditional distribution \(p_{Y^n \mid X^n}\), which extends the DMC \(p_{Y \mid X}\) to finite words of length \(n\). Note that if \(p_{Y \mid X}\) is represented by the matrix \(H\), then \(p_{Y^n \mid X^n}\) is represented by the \(n\)-th Kronecker power of \(H\), denoted by \(H^{\otimes n}\). This motivates the definition of \(\Kron : \mathbb{N}^2 \rightarrow \mathbb{N}\) which satisfies:
\begin{equation}
    \Kron(\Code_\mathcal{H}(H), n) = \Code_\mathcal{H}\left(H^{\otimes n}\right)
\end{equation}

Suppose that \(c_1 = \Code_\mathcal{H}(H_1)\) and \(c_2 = \Code_\mathcal{H}(H_2)\) for some matrices \(H_1, H_2 \in \mathcal{H}\), with dimensions \(M_1 \times N_1\) and \(M_2 \times N_2\) respectively. Denote by \(h_{1,i_1,j_1}\) and \(h_{2,i_2,j_2}\) the corresponding elements of the matrices \(H_1, H_2\). We define:
\begin{enumerate}
    \item \(h(c_1,c_2,i_1,i_2,j_1,j_2) = \Element(c_1,i_1,j_1) \cdot_{\mathbb{R}_c} \Element(c_2,i_2,j_2)\), which computes a code of the computable number \(h_{1,i_1,j_1} \cdot h_{2,i_2,j_2}\). This corresponds to the Kronecker product \(\begin{bmatrix} h_{1,i_1,j_1} \end{bmatrix} \otimes \begin{bmatrix} h_{2,i_2,j_2} \end{bmatrix}\). 


    \item \(\Row_1 (c_1,c_2,i_1,i_2,j_1)\), which computes an encoding of the Kronecker product \(\begin{bmatrix} h_{1,i_1,j_1} \end{bmatrix} \otimes \begin{bmatrix} h_{2,i_2,1} & \dots & h_{2,i_2,N_2} \end{bmatrix}\). \(\Row_1\) is given by:
    \begin{equation}
    \begin{aligned}
        \Row_1&(c_1, c_2, i_1, i_2, j_1) = \\ &R_1(\ColumnNumber(c_2), c_1, c_2, i_1, i_2, j_1)
    \end{aligned}
    \end{equation}
    where \(R_1 : \mathbb{N}^6 \rightarrow \mathbb{N}\) is defined by primitive recursion as:
    \begin{gather}
        R_1(0, c_1,c_2,i_1,i_2,j_1) = \langle \varepsilon \rangle \\
        \begin{aligned}
        R_1(y+1,c_1,c_2,i_1,&i_2,j_1) = \\
        \app(&R_1(y,c_1,c_2,i_1,i_2,j_1), \\
        &h(c_1,c_2,i_1,i_2,j_1,y+1))            
        \end{aligned}
    \end{gather}

    \item \(\Row_2(c_1,c_2,i_1,i_2)\), which encodes the Kronecker product \(\begin{bmatrix} h_{1,i_1,1} & \dots & h_{1,i_1,N_1} \end{bmatrix} \otimes \begin{bmatrix} h_{2,i_2,1} & \dots & h_{2,i_2,N_2} \end{bmatrix}\). It is defined by:
    \begin{equation}
    \begin{aligned}
        \Row_2&(c_1,c_2,i_1,i_2) = \\
        &R_2(\ColumnNumber(c_1),c_1,c_2,i_1,i_2)        
    \end{aligned}
    \end{equation}
    where \(R_2 : \mathbb{N}^5 \rightarrow \mathbb{N}\) is defined by primitive recursion as:
    \begin{gather}
        R_2(0, c_1,c_2,i_1,i_2) = \langle \varepsilon \rangle \\
        \begin{aligned}
        R_2(y&+1,c_1,c_2,i_1,i_2) = \\
        &R_2(y,c_1,c_2,i_1,i_2) * \Row_1(c_1,c_2,i_1,i_2,y+1)            
        \end{aligned}
    \end{gather}

    \item \(\Col_1 (c_1,c_2,i_1)\), which computes an encoding of the Kronecker product \(\begin{bmatrix} h_{1,i_1,1} & \dots & h_{1,i_1,N_1} \end{bmatrix} \otimes H_2\). \(\Col_1\) is given by:
    \begin{equation}
    \begin{aligned}
        \Col_1&(c_1, c_2, i_1) = \\ &C_1(\RowNumber(c_2), c_1, c_2, i_1)
    \end{aligned}
    \end{equation}
    where \(C_1 : \mathbb{N}^4 \rightarrow \mathbb{N}\) is defined by primitive recursion as:
    \begin{gather}
        C_1(0, c_1,c_2,i_1) = \langle \varepsilon \rangle \\
        \begin{aligned}
        C_1(y+1&,c_1,c_2,i_1) = \\
        &\app(C_1(y,c_1,c_2,i_1), \Row_2(c_1,c_2,i_1,y+1))            
        \end{aligned}
    \end{gather}
    \item \(\Col_2 (c_1,c_2)\), which computes an encoding of the Kronecker product \(H_1 \otimes H_2\). \(\Col_2\) is given by:
    \begin{equation}
    \begin{aligned}
        \Col_2&(c_1, c_2) = C_2(\RowNumber(c_1), c_1, c_2)
    \end{aligned}
    \end{equation}
    where \(C_2 : \mathbb{N}^3 \rightarrow \mathbb{N}\) is defined by primitive recursion as:
    \begin{gather}
        C_2(0, c_1,c_2) = \langle \varepsilon \rangle \\
        \begin{aligned}
        C_2(y+1&,c_1,c_2) = \\
        &C_2(y,c_1,c_2) * \Col_1(c_1,c_2,y+1)            
        \end{aligned}
    \end{gather}
\end{enumerate}
We can now define the primitive recursive function \(\Kron\) as \(\Kron(c, n) = K(n \dotdiv 1, c)\), where \(K : \mathbb{N}^2 \rightarrow \mathbb{N}\) is defined by:
\begin{equation}
    \begin{cases}
    K(0, c) = c \\
    K(y + 1, c) = \Col_2(K(y, c), c)        
    \end{cases}
\end{equation}

Suppose that \(c_\mathcal{C} = \Code_\mathcal{C}(\mathcal{C})\) and \(c_H = \Code_\mathcal{H}(H)\) for some block code \(\mathcal{C} = (E, D) \in \mathbf{C}\) with message set \([m]\) and for some matrix \(H \in \mathcal{H}\) representing a DMC \(p_{Y \mid X}\). We define:
\begin{enumerate}
    \item the primitive recursive function \(\CodeLength : \mathbb{N} \rightarrow \mathbb{N}\), which computes the codeword length of \(\mathcal{C}\), by:
    \begin{equation}
        \CodeLength(c_\mathcal{C}) = \lh((c_\mathcal{C})_{0,0})
    \end{equation}
    \item the primitive recursive function:
    \begin{equation}
    \begin{aligned}
        P(c_\mathcal{C}, c_H,& i, j) = \\
        &\Element(\Kron(c_H, \CodeLength(c_\mathcal{C})), i,j)
    \end{aligned}
    \end{equation}
    \(P\) calculates the probability of transition from the \(i\)-th input word to the \(j\)-th output word, when we use the block code \(\mathcal{C}\) on the DMC \(p_{Y \mid X}\).

    \item the primitive recursive function \(LO : \mathbb{N}^2 \rightarrow \mathbb{N}\), which calculates the order of a word \(\bar{x} \in [M]^n\) in the lexicographic enumeration of the elements of \([M]^n\). In particular, it satisfies:
    \begin{equation}
        LO(\langle x_0, x_1, \dots x_{n-1} \rangle, M) = 1 + \sum_{i=1}^{n}(x_{n-i} - 1) \cdot M^{i-1}
    \end{equation}
    

    

    Note that, since all codewords of a given block code \(\mathcal{C}\) have the same length \(n\), it suffices to define \(LO\) so that it computes the lexicographical order of \(\bar{x}\) over the set \([M]^n\), rather than over the commonly implied ordered set \([M]^{\leq n} = \bigcup_{0\leq i \leq n} [M]^i\). Moreover, the length \(n\) can be directly derived from the code \(\langle \bar{x} \rangle\), which means that the only required arguments for this function are \(\langle \bar{x} \rangle\) and \(M\). The function is defined as follows:
    \begin{equation}
        LO(c, M) = L(\lh(c),c,M)
    \end{equation}
    where the function \(L : \mathbb{N}^3 \rightarrow \mathbb{N}\) is defined with primitive recursion by:
    \begin{gather}
        L(0, c, M) = 1 \\
        \begin{aligned}
        L(y&+1, c, M) = \\
        & L(y, c, M) + ((c)_{\lh(c) \dotdiv (y+1)}\dotdiv 1) \cdot M^{y}
        \end{aligned}
    \end{gather}

    \item the primitive recursive relation \(W \subseteq \mathbb{N}^3\), by:
    \begin{equation}
        W(c_\mathcal{C}, i, d) \Leftrightarrow (c_\mathcal{C})_{1,i \dotdiv 1} \neq d
    \end{equation}
    \(W(c_\mathcal{C}, i, d)\) is true iff for the \(i\)-th word \(\bar{y}\) in the lexicographic enumeration of \([N]^n\) we have \(D(\bar{y}) \neq d\), where \(N\) is the size of the output alphabet \(Y\).

    \item the primitive recursive function \(\Lambda_1 : \mathbb{N}^4 \rightarrow \mathbb{N}\) by:
    \begin{gather}
        \Lambda_1(0, c_\mathcal{C}, c_H, d) = c_{0,\mathbb{R}_c} \\
        \begin{aligned}
            \Lambda_1&(i+1, c_\mathcal{C}, c_H, d) = \Lambda_1(i, c_\mathcal{C}, c_H, d) +_{\mathbb{R}_c}\\
            &P(c_\mathcal{C}, c_H, LO((c_\mathcal{C})_{0,d \dotdiv 1},\RowNumber(c_H)), i+1) \\
            & \quad\quad\quad\quad\quad\quad\quad\quad\quad\quad\quad\quad\quad\quad\cdot W(c_\mathcal{C}, i+1, d)
        \end{aligned}
    \end{gather}
    where \(c_{0,\mathbb{R}_c}\) is a code of the computable number 0. When \(d \in [m]\), \(\Lambda_1(i, c_\mathcal{C},c_H,d)\) calculates the first \(i\) terms of the sum:
    \begin{equation}\label{eq:sum}
        \lambda(\mathcal{C}, p_{Y \mid X}, d) = \sum_{\bar{y} \in [N]^n : D(\bar{y}) \neq d} p_{Y^n \mid X^n}(\bar{y} \mid E(d))
    \end{equation}
    which represents the block error probability for a specific message \(d \in \mathcal{M}\).

    \item the primitive recursive function:
    \begin{equation}
    \begin{aligned}
        \Lambda_2&(c_\mathcal{C},c_H,d) = \\
        &\Lambda_1 \left( \ColumnNumber(c_H)^{\CodeLength(c_\mathcal{C})} , c_\mathcal{C},c_H,d \right)
    \end{aligned}
    \end{equation}
    which calculates the whole sum of equation \eqref{eq:sum}.

    \item the primitive recursive function \(\Lambda_3 : \mathbb{N}^3 \rightarrow \mathbb{N}\) by:
    \begin{gather}
        \Lambda_3(0, c_\mathcal{C}, c_H) = c_{0, \mathbb{R}_c} \\
        \begin{aligned}
        \Lambda_3(i+1, &c_{\mathcal{C}}, c_H) =\\
        &\max_{\mathbb{R}_c}(\Lambda_3(i, c_\mathcal{C}, c_H), \Lambda_2(c_\mathcal{C},c_H,i+1))
        \end{aligned}
    \end{gather}
    When \(i \in [m]\) \(\Lambda_3(i, c_\mathcal{C}, c_H)\) calculates the maximum:
    \begin{equation}
        \max \{ \lambda(\mathcal{C}, p_{Y \mid X}, d) \mid 0 < d \leq i \}
    \end{equation}

\end{enumerate}

The function \(\Lambda\) can finally be defined as:
\begin{equation}
    \Lambda(c_\mathcal{C}, c_H) = \Lambda_3( \lh((c_\mathcal{C})_0), c_\mathcal{C}, c_H)
\end{equation}

\subsection{Defining the Function \(\AchievesError\)}

In this subsection, we define a recursive function \(\AchievesError : \mathbb{N}^3 \rightarrow \mathbb{N}\) that determines whether there exists a block code in a given list \(( \mathcal{C}_0, \mathcal{C}_1, \dots, \mathcal{C}_{k-1} )\) which achieves a desired error tolerance \(\epsilon = 2^{-b}\) for a given DMC \(p_{Y \mid X}\). This function, together with \(\Codes(M, N, m, n)\), allows us to verify whether a specific codeword length \(n\) suffices to achieve the desired error \(2^{-b}\), or whether \(n\) must be increased according to Algorithm \ref{alg:algorithm_working}.

The inputs to \(\AchievesError\) are the encoding \(c =\langle c_0, c_1, \dots c_{k-1} \rangle\), where \(c_i = \Code_\mathbf{C}(\mathcal{C}_i)\), an encoding \(c_H\) of the DMC and the number \(b\). Let \(c_{\lambda,i} = \Lambda(c_i, c_H)\) be the code of the computable real \(\lambda_{\max}(\mathcal{C}_i, p_{Y \mid X})\), as computed by the function \(\Lambda\). If there exists a block code \(\mathcal{C}_i\) in the list that achieves the condition \(\rat\left(\varphi^1\left(c_{\lambda,i}, b+2\right)\right) < 2^{-b-1}\) of Algorithm \ref{alg:algorithm_working}, then \(\AchievesError\) returns the index \(i\) of the first such code. If no such code exists, then \(\AchievesError\) returns \(k\). More precisely, when \(c = \langle c_0, c_1, \dots c_{k-1} \rangle\), with \(\isCode_{\mathbf{C}}(c_i), \ \forall i\) we have:
\begin{equation}
\begin{aligned}
    &\AchievesError(c,c_H,b) = \\
    &\quad\begin{cases}
        \min \left\{i < k \middle| \rat\left(\varphi^1\left(c_{\lambda,i}, b+2\right)\right) < 2^{-b-1}\right\}, &\text{if such } i \\&\text{exists} \\
        k, &\text{otherwise}
    \end{cases}
\end{aligned}
\end{equation}
This function can be defined by:
\begin{equation}
\begin{aligned}
    &\AchievesError(c, c_H, b) = \\
    &\quad\mu i \leq \lh(c)\dotdiv 1 \left(\varphi^1\left(\Lambda((c)_i,c_H), b+2\right) <_\mathbb{Q} \left\langle 0, 1, 2^{b+1} \right\rangle\right)
\end{aligned}
\end{equation}

\subsection{Defining the Function \(\MessageNumber\)}

We proceed to defining \(\MessageNumber\), with
\begin{equation}
    \MessageNumber(c_R, n) = \left\lceil 2^{n\rat(c_R)} \right\rceil
\end{equation}
whenever \(\isRat(c_R)\) and \(\rat(c_R) > 0\). Let \(c_R \in \mathbb{N}\) be such a number. This function computes the cardinality of the message set \(\mathcal{M}\) that a block code \(\mathcal{C}\) with codeword length \(n\) must have in order to achieve a coding rate \(R\). We have:
\begin{align}
    \left\lceil 2^{n\rat(c_R)} \right\rceil &= \min\left\{i \in \mathbb{N} \middle| 2^{n  \frac{N_\mathbb{Q}(c_R)}{D_\mathbb{Q}(c_R)}} \leq i\right\} \\
    &= \min\left\{i \in \mathbb{N} \middle| 2^{n  N_\mathbb{Q}(c_R)} \leq i^{D_\mathbb{Q}(c_R)}\right\} \label{eq:message_number}
\end{align}

Since \(D_\mathbb{Q}(c_R) \geq 1\), the inequality \(2^{n  N_\mathbb{Q}(c_R)} \leq i^{D_\mathbb{Q}(c_R)}\) is true for \(i = 2^{nN_\mathbb{Q}(c_R)}\). Therefore, equation \eqref{eq:message_number} can be expressed with the bounded minimization:
\begin{equation}
\begin{aligned}
&\MessageNumber(n,c_R) = \\
&\quad\quad\quad\quad\mu i\leq 2^{n N_\mathbb{Q}(c_R)}: \left(2^{n N_\mathbb{Q}(c_R)} \leq i^{D_\mathbb{Q}(c_R)} \right)
\end{aligned}
\end{equation}

\subsection{The recursive function \(\FindCode\)}

We will now combine the previously defined functions to construct the recursive function \(\FindCode\). When this function takes as input an encoding of a DMC \(p_{Y \mid X}\) with \(C(p_{Y \mid X}) \neq 0\), an encoding of a positive rate \(R < C(p_{Y \mid X})\) and an encoding of a rational error tolerance \(\epsilon > 0\), it returns an encoding of a \(\left(\left\lceil 2^{nR} \right\rceil, n\right)\) block code \(\mathcal{C}\) with \(\lambda_{\max}(\mathcal{C}, p_{Y \mid X}) < \epsilon\).

We first define \(C : \mathbb{N}^3 \rightarrow \mathbb{N}\) by:
\begin{equation}
\begin{aligned}
    C(c_H, c_R, n) = \\
    \Codes(&\RowNumber(c_H), \ColumnNumber(c_H),\\
     &\MessageNumber(n, c_R), n)
\end{aligned}
\end{equation}
\(C(c_H, c_R, n)\) is the encoding of the sequence of all \(\left( \left\lceil 2^{n\rat(c_R)}\right\rceil ,n \right) \) block codes for the DMC encoded by \(c_H\).

We find the minimum required codeword length for achieving the target error probability \(\epsilon\) by using a recursive function \(\MinLength : \mathbb{N}^3 \rightharpoonup \mathbb{N}\) defined by:
\begin{equation}
\begin{aligned}
    \MinLength(c_H,& c_R, c_\epsilon) = \\
    \mu i : (&\AchievesError(C(c_H, c_R, i), c_H, \BLB(c_\epsilon)) \\
    &< \lh(C(c_H, c_R, i)))
\end{aligned}
\end{equation}

The function \(\FindCode\) can finally be defined by:
\begin{equation}
\begin{aligned}
    \FindCode(c_H, c_R, c_\epsilon) = \\
    \AchievesError(&C(c_H, c_R, \MinLength(c_H, c_R, c_\epsilon)),\\
    &c_H, \BLB(c_\epsilon))
\end{aligned}
\end{equation}

The function \(\FindCode\) provides a general solution to the problem of finding capacity-achieving codes. Given as input encodings of a DMC \(p_{Y \mid X}\), a target rate \(R < C(p_{Y \mid X})\) and a rational error threshold \(\epsilon > 0\), the function returns an encoding of a block code \(\mathcal{C}\) for \(p_{Y \mid X}\) that achieves a rate of at least \(R\) and satisfies \(\lambda_{\max}\left(\mathcal{C}, p_{Y \mid X}\right)\). Moreover, \(\FindCode\) is recursive, as it is constructed by combining recursive functions using operations that preserve recursiveness.

\subsection{Generalization for \(\epsilon \in \mathbb{R}_c\)}\label{sec:extension}

The function can be naturally extended to handle error thresholds \(\epsilon \in \mathbb{R}_c\). This is accomplished by defining a recursive function \(\RLB : \mathbb{N} \rightharpoonup \mathbb{N}\) (acronym for Rational Lower Bound), which, given a code \(c_\epsilon\) of a computable real number \(\epsilon > 0\), returns a code of a positive rational lower bound of \(\epsilon\). That is, \(\isRat(\RLB(c_\epsilon))\) holds, and \(\rat(\RLB(c_\epsilon)) = q \in \mathbb{Q}\) for some \(0 < q < \epsilon\).

\begin{lemma}
    There exists a \(\mu\)-recursive function \(\RLB : \mathbb{N} \
    \rightharpoonup \mathbb{N}\) such that, if \(\isCom(c_\epsilon)\) and \(\com(c_\epsilon) = \epsilon > 0\), then:
    \begin{gather}
        \RLB(c_\epsilon) \downarrow \\
        \isRat(\RLB(c_\epsilon)) \\
        0 < \rat(\RLB(c_\epsilon)) < \epsilon
    \end{gather}
\end{lemma}

\begin{proof}
    Let \(c_\epsilon \in \mathbb{N}\) be as described above. Set \(f_\epsilon (n) = \varphi^1(c_\epsilon, n)\). By definition \ref{def:computable} we have:
    \begin{equation}
        \left| \rat(f_\epsilon(n)) - \epsilon \right| < \frac{1}{2^n}, \quad \forall n \in \mathbb{N}
    \end{equation}
This implies that for all \(n \in \mathbb{N}\):
\begin{equation}\label{eq:RLB_inequality}
    \epsilon - \frac{1}{2^{n-1}} < \rat(f_\epsilon(n)) - \frac{1}{2^n} < \epsilon
\end{equation}
Thus, \(\rat(f_\epsilon(n)) - \frac{1}{2^n}\) is always a rational lower bound of \(\epsilon\). In addition, when \(n > \log(1/\epsilon) + 1\), the expression \(\epsilon - \frac{1}{2^{n-1}}\) is strictly positive, and by inequality \eqref{eq:RLB_inequality} we have:
\begin{equation}
    \rat(f_\epsilon(n)) - \frac{1}{2^n} > 0
\end{equation}

Therefore, there exist \(n \in \mathbb{N}\) for which \(\rat(f_\epsilon(n)) - \frac{1}{2^n}\) is positive. Hence, the minimization:
\begin{equation}
    L(c_\epsilon) = \mu n:\left( \varphi^1(c_\epsilon, n) -_\mathbb{Q} \langle 0, 1, 2^n  \rangle >_\mathbb{Q} \langle 0, 0, 1 \rangle \right)
\end{equation}
converges, and it returns a number \(n = L(c_\epsilon)\) for which \(\rat(f_\epsilon(n)) - \frac{1}{2^n}\) is a positive rational lower bound of \(\epsilon\). Consequently, \(\RLB\) can be defined by:
\begin{equation}
    \RLB(c_\epsilon) = \varphi^1(c_\epsilon, L(c_\epsilon)) -_\mathbb{Q} \left\langle 0, 1, 2^{L(c_
    \epsilon)} \right\rangle
\end{equation}

\end{proof}

Having the recursive function \(\RLB\), the extension of \(\FindCode\) that works for \(\epsilon \in \mathbb{R}_c\) is defined by:
\begin{equation}
    \FindCodeExt(c_H, c_R, c_\epsilon) = \FindCode(c_H, c_R, \RLB(c_\epsilon))
\end{equation}

If \(c_H\) is an encoding of a DMC \(p_{Y \mid X}\), \(c_R\) is a code of a rational \(R < C(p_{Y \mid X})\) and \(c_\epsilon\) is a code of a computable real \(\epsilon > 0\), then \(\RLB(c_\epsilon)\) is a code of a rational \(q\) with \(0 < q < \epsilon\) and \(\FindCodeExt(c_H, c_R, c_\epsilon)\) is an encoding of a block code \(\mathcal{C}\) for \(p_{Y \mid X}\) with rate at least \(R\) and with maximum block error probability:
\begin{equation}
    \lambda_{\max}\left(\mathcal{C}, p_{Y \mid X}\right) < q < \epsilon
\end{equation}

\subsection{Extension to \(R \in \mathbb{R}_c\)}\label{sec:extension2}

A similar extension can be considered for the case where \(R \in \mathbb{R}_c\). A natural attempt would be to define a function \(\MessageNumberExt(n, c_R) = \left\lceil 2^{nR} \right\rceil\), assuming \(c_R\) encodes a computable real number \(R\). However, such a function cannot be defined recursively.

\begin{lemma}
There does not exist any \(\mu\)-recursive function \(\MessageNumberExt : \mathbb{N}^2 \rightharpoonup \mathbb{N}\) satisfying:
\begin{equation}
\MessageNumberExt(n, c_R) = \left\lceil 2^{nR} \right\rceil
\end{equation}
for all \(n, c_R \in \mathbb{N}\) and \(R \in \mathbb{R}_c\) such that \(\isCom(c_R)\) and \(\com(c_R) = R > 0\).
\end{lemma}

\begin{proof}
    By way of contradiction, suppose that there exists a recursive function \(\MessageNumberExt\) with the above properties. Let \(c_M\) be defined as in the proof of Lemma \ref{lemma:non_computable_relations} and let \(c_1\) be a code of the computable real number 1. From expressions \eqref{eq:com_equality} and \eqref{eq:com_ordering} we have:
    \begin{gather}
        \com(c_1 +_{\mathbb{R}_c} S^2_1(c_M,e,x)) = 1 \Leftrightarrow \neg H(e,x) \\
        \com(c_1 +_{\mathbb{R}_c} S^2_1(c_M,e,x)) > 1 \Leftrightarrow H(e,x)
    \end{gather}
    The above equivalences imply that:
    \begin{align}
         &H(e,x)\Leftrightarrow\left\lceil 2^{\com(c_1 +_{\mathbb{R}_c} S^2_1(c_M,e,x))} \right\rceil = 2 \\
         &\Leftrightarrow \MessageNumberExt(1,\com(c_1 +_{\mathbb{R}_c} S^2_1(c_M,e,x))) = 2
    \end{align}
    from which it follows that the relation \(H\) is recursive, which is a contradiction. Therefore, such a function \(\MessageNumberExt\) does not exist.
    
\end{proof}

This motivates an approach analogous to the one outlined in Subsection \ref{sec:extension}, where the goal is to find a rational number \(\hat{R}\) such that \(R \leq \hat{R} < C(p_{Y \mid X})\), which can then be used as input to \(\FindCodeExt\). Achieving this requires the definition of two additional \(\mu\)-recursive functions:
\begin{enumerate}
    \item A function \(\Capacity : \mathbb{N} \rightharpoonup \mathbb{N}\), which takes as input an encoding of a DMC \(p_{Y \mid X}\) and returns a code of its capacity \(C(p_{Y \mid X})\). Computing the capacity of a DMC reduces to solving a convex optimization problem, which can be approached using methods such as the steepest descent algorithm. However, formally encoding such a procedure as a recursive function is beyond the scope of this paper.

    \item A function \(\RatInterpolation : \mathbb{N}^2 \rightharpoonup \mathbb{N}\), which takes as input two codes \(c_\alpha\) and \(c_\beta\) of computable real numbers \(\alpha\) and \(\beta\) with \(\alpha < \beta\), and returns a code of a rational number \(q\) satisfying \(\alpha < q < \beta\). The construction of this function is provided in Lemma \ref{lemma:interpolation}.
\end{enumerate}

\begin{lemma}\label{lemma:interpolation}
    There exists a recursive function \(\RatInterpolation : \mathbb{N}^2 \rightharpoonup \mathbb{N}\) satisfying:
    \begin{gather}
        \isRat(\RatInterpolation(c_\alpha, c_\beta)) \\
        \com(c_\alpha) < \RatInterpolation(c_\alpha, c_\beta) < \com(c_\beta)
    \end{gather}
    for all \(c_\alpha, c_\beta \in \mathbb{N}\) such that \(\isCom(c_\alpha)\), \(\isCom(c_\beta)\) and \(\com(c_\alpha) < 
    \com(c_\beta)\).
\end{lemma}

\begin{proof}
    Let \(c_\alpha\), \(c_\beta\) be as above and set \(\alpha = \com(c_\alpha)\), \(\beta = \com(c_\beta)\). Define \(f_\alpha(n) = \varphi^1(c_\alpha, n)\), \(f_\beta(n) = \varphi^1(c_\beta, n)\) and \(g(n) = (f_\alpha(n) +_\mathbb{Q} f_\beta(n)) /_\mathbb{Q} \langle 0, 2, 1 \rangle\). We will show that if:
    \begin{equation}\label{eq:rat_int_1}
    \rat(f_\beta(n)) - \rat(f_\alpha(n)) > \frac{1}{2^{n-1}}\
    \end{equation}
    then the value \(g(n)\) satisfies the desired inequality \(\alpha < \rat(g(n)) < \beta\). Suppose that the inequality \eqref{eq:rat_int_1} is true. Then, since \(\rat(g(n)) = \frac{\rat(f_\alpha(n)) + \rat(f_\beta(n))}{2}\), we have:
    \begin{align}
        \alpha < \rat(f_\alpha(n)) + \frac{1}{2^n} < \rat(g(n)) < \rat(f_\beta(n)) - \frac{1}{2^n} < \beta
    \end{align}
    
    Therefore, if the minimization:
    \begin{align}
        N(c_\alpha, c_\beta) &= \mu n: \left(f_\alpha(n) -_\mathbb{Q} f_\beta(n) >_\mathbb{Q} \left\langle 0, 1, 2^{n-1} \right\rangle\right) \\
        &= \mu n: \left(\varphi^1(c_\alpha, n) -_\mathbb{Q} \varphi^1(c_\beta, n) >_\mathbb{Q} \left\langle 0, 1, 2^{n-1} \right\rangle\right) \label{eq:rat_int_2}
    \end{align}
    converges, then the desired function \(\RatInterpolation\) can be defined by:
    \begin{equation}
        \RatInterpolation(c_\alpha, c_\beta) = g(N(c_\alpha, c_\beta))
    \end{equation}
    It now suffices to show that the minimization \eqref{eq:rat_int_2} converges. Set \(\epsilon = \beta - \alpha > 0\). Note that for \(n > \log(1/\epsilon) + 2 \Leftrightarrow \epsilon- \frac{1}{2^{n-1}} > \frac{1}{2^{n-1}}\) we have:
    \begin{align}
        \rat(f_\beta(n)) - \rat(f_\alpha(n)) &> \beta - \frac{1}{2^n} - \left( \alpha + \frac{1}{2^n} \right) \\
        &= \epsilon - \frac{1}{2^{n-1}} > \frac{1}{2^{n-1}}
    \end{align}
    From the above it follows that \eqref{eq:rat_int_2} converges. This concludes the proof.
\end{proof}

\section{Conclusion}\label{sec:conclusion}

This work has established the existence of a Turing machine that solves the problem of constructing capacity-achieving codes. This machine takes as input a DMC \(p_{Y \mid X}\), an error tolerance \(\epsilon\) and a coding rate \(R\), and in the case where \(R < C(p_{Y \mid X})\), it outputs a block code \(\mathcal{C}\) for \(p_{Y \mid X}\) with rate at least \(R\) and \(\lambda_{\max}(\mathcal{C}, p_{Y \mid X}) < \epsilon\). The construction works for the general case where all the transition probabilities of \(p_{Y \mid X}\) and the tolerance \(\epsilon\) are computable real numbers, and the rate \(R\) is a rational number. Furthermore, we discussed a generalization of this machine that works for \(R \in \mathbb{R}_c\). These results demonstrate that there exist general algorithmic methods for constructing capacity-achieving codes that work for all DMCs.

While the proposed machines do solve the general problem, they rely on exhaustive search techniques and exhibit exponential complexity, rendering them impractical. Nevertheless, several refinements can be applied to slightly reduce the complexity of the resulting algorithms. First, the exponential overhead of the classical encoding \(\langle \rangle\) can be avoided either by employing alternative polynomial-time primitive recursive encodings, such as extentions of the Cantor pairing function discussed in \cite{odifreddi1992classical}, or by dispensing with such encodings altogether and instead implementing data structures that support efficient list operations. Second, the search space of block codes \(\Codes(M, N, m, n)\) can be significantly reduced by restricting attention to codes \(\mathcal{C} = (E, D)\) in which the encoding function \(E : \mathcal{M} \rightarrow X^n\) is injective. It is straightforward to show that any non-injective encoding function yields maximum error probability at least \(\frac{1}{2}\); hence, since all relevant cases satisfy \(\epsilon < \frac{1}{2}\), no generality is lost. Furthermore, if \(\mathcal{M} = \{s_1, s_2, \dots s_m\}\), we may assume without loss of generality that the codewords \(E(s_1), E(s_2), \dots E(s_m)\) are ordered lexicographically, reducing the search space by an additional factor of \(m!\). Finally, for many practical channels---including the binary symmetric channel and the binary erasure channel with crossover or erasure probability less than \(\frac{1}{2}\)---the optimal decoding function \(D : Y^n \rightarrow \mathcal{M}\) satisfies the minimum Hamming distance rule:
\begin{equation}
D(\bar{y}) \in \arg\min_{s \in \mathcal{M}} d_H(\bar{y}, E(s)),
\end{equation}
where \(d_H\) denotes the Hamming distance. Thus, it suffices to restrict attention to codes consistent with this property. Nonetheless, even with all of these reductions, the search space remains exponential, and the resulting algorithms remain computationally infeasible.

The exponential nature of the algorithms does not diminish the significance of this work, since our goal is to establish the existence of a universal code-construction method, rather than to propose an efficient one. Having rigorously demonstrated the existence of such a method, future work may focus on identifying more efficient approaches or on characterizing possible trade-offs between algorithmic complexity and the generality of the channel models to which the method applies. For instance, although the consideration of computable real parameters is of theoretical interest, it is unnecessary in practice, since by the density of \(\mathbb{Q}\) in \(\mathbb{R}\), any channel can be approximated to arbitrary precision using a transition matrix with rational entries. Another natural direction is the algorithmic study of more general channel models, such as finite-state channels (FSCs).

\appendices

\section{Discussion for Capacity-Achieving Sequences of Codes}\label{app:extension3}

In this appendix, we compare the framework adopted in this work with that presented in \cite{boche2022turing}. The framework of \cite{boche2022turing} considers Turing machines that take as input a DMC and a parameter \(n\), and output a block code of length \(n\), whose error probability tends to zero and rate tends to the channel capacity as \(n \to \infty\). In contrast, our machine takes as input a DMC together with explicit bounds on the code rate and error probability. Our approach, however, can be adapted to align with the framework of \cite{boche2022turing} by providing our machine with input parameters \(\epsilon\) and \(R\) that asymptotically approach \(0\) and \(C(p_{Y \mid X})\), respectively. Specifically, we describe the construction of a Turing machine that generates a sequence of block codes \(\{\mathcal{C}_k\}\) for a given DMC \(p_{Y \mid X}\), such that, as \(k \to \infty\), the rate of \(\mathcal{C}_k\) approaches the capacity \(C(p_{Y \mid X})\) and its error probability tends to zero. This machine takes as input encodings of a DMC \(p_{Y \mid X}\) with computable transition probabilities and a parameter \(k \in \mathbb{N}\), and operates as follows:
\begin{enumerate}
    \item It calculates the capacity of \(p_{Y \mid X}\) using the function \(\Capacity\) discussed in Subsection \ref{sec:extension2}.
    \item It computes the computable real number \(C(p_{Y \mid X}) - \frac{1}{k}\).
    \item It employs the function \(\RatInterpolation\) to determine a rational rate \(R\) satisfying \(C(p_{Y \mid X}) - \frac{1}{k} < R < C(p_{Y \mid X})\).
    \item It calls the function \(\FindCode\) with inputs \(p_{Y \mid X}\), the rate \(R\), and an error tolerance \(\epsilon_k = \frac{1}{k}\), and returns the resulting block code.
\end{enumerate}

Since \(R < C(p_{Y \mid X})\), the proposed machine halts and outputs a block code \(\mathcal{C}_k\) with rate \(R_k \geq R > C(p_{Y \mid X}) - \frac{1}{k}\) and maximum block error probability \(\lambda_{\max}(\mathcal{C}_k, p_{Y \mid X}) < \frac{1}{k}\). It follows that the sequence \(\{\mathcal{C}_k\}\) satisfies
\begin{gather}
\lim_{k \to \infty} R_k = C(p_{Y \mid X}) \\
\lim_{k \to \infty} \lambda_{\max}(\mathcal{C}_k, p_{Y \mid X}) = 0
\end{gather}

A key difference between the machine described above and the one proven impossible in \cite{boche2022turing} is that the latter requires the constructed code to have a block length exactly equal to \(n\). In contrast, in our construction the parameter \(k\) is not directly tied to the codeword length. In fact, within the scope of this work, there are no evident bounds on the block length as a function of the parameter \(k\). Furthermore, we employ a different formulation for describing input DMCs. In \cite{boche2022turing} the notion of a \emph{computable family of channels} is used, which can represent a broader class of DMCs than those with merely computable real parameters. However, this formulation requires that the input space of DMCs have fixed input and output alphabet sizes \(M\) and \(N\), and that the corresponding family of channels be expressible as a \(M \times N\) matrix of \emph{computable continuous} functions. In contrast, our formulation imposes no such restrictions, allowing arbitrary DMCs of any dimension to be provided as input to the same Turing machine.

\bibliographystyle{IEEEtran}
\bibliography{bib}

\end{document}